\documentclass[11pt]{article}
\usepackage[active]{srcltx}
\usepackage[T1]{fontenc}
\usepackage[utf8]{inputenc}
\usepackage{lmodern}
\usepackage[a4paper]{geometry}
\usepackage{amsmath,amssymb,amsthm,mathabx}
\usepackage{pgf,tikz}
\usepackage{mathrsfs}
\usepackage{graphicx}
\usepackage{subcaption}
\usepackage{mwe}
\usepackage[colorlinks=true,breaklinks=true,linkcolor=blue,urlcolor=green,citecolor=red]{hyperref} 
\usepackage{yfonts}
\usetikzlibrary{arrows}

\DeclareMathOperator{\ad}{ad}

\newcommand{\te}{\theta}

\newcommand{\R}{\mathbb{R}}
\renewcommand{\epsilon}{\varepsilon}

\newcommand{\p}{\partial}
\newtheorem{theorem}{Theorem}

\newtheorem{proposition}{Proposition}
\newtheorem{corollary}{Corollary}
\theoremstyle{remark}

\newtheorem{remark}{Remark}
\newtheorem{definition}{Definition}

\newcommand{\s}{\sigma}
\renewcommand{\l}{\lambda}

\renewcommand{\R}{\mathbb{R}}

\usepackage{soul}

\title{Optimal control of oscillatory neuronal models with applications to communication through coherence }
\author{Michael Orieux$^{1}$, Antoni Guillamon$^{1,2,3}$, Gemma Huguet$^{1,2,3}$ \\
\parbox{12.5cm}{
  \small
  \begin{itemize}
  \item[$^1$]
    Departament de Matem\`atiques, Universitat Polit\`ecnica de Catalunya, Barcelona, Spain 
  \item[$^2$]
    Institut de Matem\`atiques de la UPC - Barcelona Tech (IMTech), Barcelona, Spain 
  \item[$^ 3$]
  Centre de Recerca Matem\`atica, Barcelona, Spain
   \end{itemize}
 }}

\begin{document}

\maketitle



\abstract{
Macroscopic oscillations in the brain are involved in various cognitive and physiological processes, yet their precise function is not not completely understood. Communication Through Coherence (CTC) theory proposes that these rhythmic electrical patterns might serve to regulate the information flow between neural populations. Thus, to communicate effectively, neural populations must synchronize their oscillatory activity, ensuring that input volleys from the presynaptic population reach the postsynaptic one at its maximum phase of excitability. We consider an Excitatory-Inhibitory (E-I) network whose macroscopic activity is described by an exact mean-field model. The E-I network receives periodic inputs from either one or two external sources, for which effective communication will not be achieved in the absence of control. We explore strategies based on optimal control theory for phase-amplitude dynamics to design a control that sets the target population in the optimal phase to synchronize its activity with a specific presynaptic input signal and establish communication. The control mechanism resembles the role of a higher cortical area in the context of selective attention. To design the control, we use the phase-amplitude reduction of a limit cycle and leverage recent developments in this field in order to find the most effective control strategy regarding a defined cost function. Furthermore, we present results that guarantee the local controllability of the system close to the limit cycle.
}
\section{Introduction}

Macroscopic oscillations, spanning a frequency range from a few to a hundred hertz, are frequently observed in the brain \cite{Buzsaki06}, but their role is not completely understood. Communication Through Coherence (CTC) theory \cite{Fries05, Fries15}, suggests a functional role for oscillations in the context of information transmission. This theory postulates that synchronization plays a pivotal role in enhancing communication between neuronal groups.

According to CTC theory, communication between two neuronal groups is established when the oscillatory input from the presynaptic group (the sender) reaches the postsynaptic group (the receiver) at its maximum phase of excitability. This mechanism can effectively implement selective attention  \cite{Fries01, Fries08, Schoffelen11, Bosman12}. The primary goal of selective attention is to transmit the information related to the stimulus that an individual is consciously attending to. To achieve this, the oscillatory activity of the postsynaptic group needs to be synchronized with the input from the presynaptic group that codes for the attended stimulus. Simultaneously, selective attention involves the suppression of irrelevant or distracting inputs. In the context of CTC theory, this means that the coordination between pre and postsynaptic groups should be such that oscillatory inputs from other, non-attended, presynaptic groups are effectively suppressed.

Communication between populations of neurons involved in selective attention is believed to be regulated by a top-down mechanism \cite{Engel01, Hoptopdown}, that is, a feedback signal from higher-level brain regions that modulates the processing of signals in lower-level areas. These feedback signals carry information related to an individual's attentional focus. In the prospect of the article, the top-down mechanism will be represented by a control term, which corresponds to a signal that regulates the information flow when several presynaptic inputs converge to a common postsynaptic neural group. In particular, the control ensures that the oscillatory behavior of the postsynaptic group aligns appropriately with the input signal from the presynaptic group which encodes the attended stimulus.

In this paper, we present a theoretical study based on optimal-control theory and phase dynamics to explore whether an top-down input signal can set the receiving population in the optimal phase for communication with the sender.

We consider a spiking network of excitatory and inhibitory cells (E-I network), whose macroscopic activity,  characterized by average firing rates and membrane potentials, can be exactly captured using the low-dimensional mean-field models introduced in \cite{Montbrio1, Dumont19}. The E-I network, modeling the postsynaptic group, shows macroscopic oscillations in the gamma range. We perturb it with periodic inputs from different presynaptic neuronal groups encoding different stimuli. In a previous study \cite{ReynerHuguet22}, we observed that presynpatic inputs with a higher frequency than the intrinsic network gamma cycle have an advantage to entrain the target network and communicate effectively. In this new study, we develop an optimal-control strategy to set the oscillatory activity of the postsynaptic E-I network in the proper phase for communication with a particular presynaptic group. The interesting result is that our strategy applies to those cases in which the target neuronal group is oscillating in a regime that is not suitable for communication \cite{ReynerHuguet22}.

In order to establish the communication paradigm we use the phase-amplitude reduction \cite{CGH13} and apply optimal-control techniques to this framework \cite{Moehlis}. To this end, we first present novel results that guarantee the controllability of systems close to a limit cycle (Proposition~\ref{LocalCtrl}), which provides a solid basis for addressing the control problem. We adopt a Hamiltonian formulation for the optimal-control problem based on Pontryagin's Maximum Principle \cite{Agrachev04}, in contrast to the Lagrangian formulation \cite{MoehlisT} or the Hamilton-Jacobi-Bellman approach \cite{Moehlis2013}. When incorporating the phase-amplitude reduction to the optimal-control problem \cite{Moehlis}, we discuss different strategies and formulate an accurate description of the dynamics along the dominant contracting direction by taking advantage of the application of the parameterization method \cite{CFL05, GH09, Perez_Cervera20}. Finally, we apply the latter one to solve the control problem for the CTC setting.

The paper is organized as follows. In Section \ref{sec:control} we set the mathematical formalism for the control problem. In Section \ref{sec:control_LC} we discuss general results on the controllability of systems close to a limit cycle. In 
Section \ref{sec:PAreduc} we discuss the mathematical formalism for the phase-amplitude reduction in the context of control theory and in Section \ref{sec:Ocontrol} we discuss the numerical implementation details. Finally, in Section \ref{sec:CTC} we present the main results in which we apply the control techniques discussed in the previous sections to the CTC problem. We end with a discussion in Section~\ref{sec:discussion}. The Appendix illustrates the application of our results on local controllability to other models in neuroscience beyond the mean-field models used in the main text.

\section{Control theory for control-affine systems}\label{sec:control}

In this section we set the background on control theory that will be used along this manuscript.

Let $f:\mathbb{R}^n\times U\rightarrow \mathbb{R}^n$ be a smooth function where $U\subset\mathbb{R}^m$ is the \emph{control set}. 
The control system writes as 
\begin{equation}
    \begin{cases}
    \dot{x}=f(x,u),\\
    x(0)=x_0, \\
    x(t_f)=x_f.
    \end{cases}
    \label{CS}
\end{equation}
where $u\in L^\infty([0,t_f],U)$ is the control. We denote by $x_u(t,x_0)$ the flow of $f(\cdot,u)$ at time $t$ from $x_0$.

\begin{definition}
The \emph{reachable set for \eqref{CS} from $x_0$ at time $t$} is defined by
\begin{equation}
    \mathcal{A}(x_0,t)=\{x_1 \in \mathbb{R}^n \; | \, \; \exists u\in L^\infty([0,t_f],U)\text{ with }x_u(t;x_0)=x_1\};
\end{equation}
the \emph{reachable set for \eqref{CS} from $x_0$} is
\begin{equation}
    \mathcal{A}(x_0)=\bigcup_{t\geq0}\mathcal{A}(x_0,t).
\end{equation}
\end{definition}
We say that \eqref{CS} is \emph{controllable from $x_0$} if $\mathcal{A}(x_0)=\mathbb {R}^n$, and \emph{controllable} if the latter is true for every $x_0\in\R^n.$
On the other hand, \emph{local} controllability around $x_0$ means that $x_0$ belongs to the interior of  $\mathcal{A}(x_0)$.

When dealing with local controllability of nonlinear systems, we will use the following classical result (see, for instance, \cite{Coron07}):

 \begin{theorem}
             Let $(\bar{x},\bar{u})$ be a solution of the control system \eqref{CS}. If the linearized system along $(\bar{x},\bar{u})$ is controllable, then the nonlinear system is locally controllable along the trajectory $x$ with any control $u$ close to $\bar{u}$. That is, for all $\epsilon>0$ there exists $\eta>0$ such that for all $a,b\in\R^n$, with $|\bar{x}(0)-a|+|\bar{x}(t_f)-b|<\eta$, there is a solution $(x,u)$ with $\|u-\bar{u}\|_\infty<\epsilon$ satisfying $x(0)=a$, $x(t_f)=b$.
         \label{localcont}
\end{theorem}

We will also use the following theorem by Chang (see \cite{Chang}) that establishes a criterion of controllability for time-dependent linear systems. 
  \begin{theorem}
  Let $A:[0,t_f]\mapsto M_{n}(\R)$ and $B:[0,t_f]\mapsto M_{n,m}(\R)$ be smooth matrices. Let us define  $B_i(t)$ recursively by $B_0=B$ and $B_i=\dot{B}_{i-1}-AB_{i-1}$, where $\dot{B}$ indicates the derivative with respect to time.  Then, if 
\begin{equation}
    \mathrm{span}\{B_i(t) \, \tilde{u},\; \tilde{u}\in\R^m \}_{i\geq 0}=\R^n
    \label{lincrit}
\end{equation} 
holds for any $t \in[0,t_f]$,
the linear system $\dot{x}=Ax+Bu$ is controllable on $[0,t_f]$.
\label{lincont}
  \end{theorem}
  Most examples in nature, as well as the ones encountered in this paper, are \emph{control-affine} systems, that is, $f(x,u)=F_0(x)+\sum_{i=1}^m u_iF_i(x)$ in \eqref{CS}, with $(F_i)_{i=0,\dots,m}$ being a family of vector fields on $\R^n$ and $u=(u_1,\dots,u_m)$, with $u_i\in L^\infty([0,t_f],U)$ for $i=1 \ldots m$,
\begin{equation}
    \begin{cases}
    \dot{x}=F_0(x)+\sum_{i=1}^{m}u_iF_i(x), \\
    x(0)=x_0,\\
    x(t_f) = x_f,
    \end{cases}
    \label{AC}
\end{equation}
with $x_0, x_f \in \mathbb{R}^n$.

Let $\varphi: \mathbb{R}^n \times U\rightarrow\mathbb{R}$ be a smooth function and consider system (\ref{AC}), with the extra condition of minimizing a cost function $C(x,u)=\int_{0}^{t_f}\varphi(x(t),u(t))dt$ along the trajectory. Therefore, the control system writes as
 \begin{equation}
    \begin{cases}
    \dot{x}=F_0(x)+\sum_{i=1}^{m}u_iF_i(x), \\
    x(0)=x_0,\\
    x(t_f) = x_f,\\
    \int_{0}^{t_f}\varphi(x(t),u(t))dt\rightarrow \min.
    \end{cases}
    \label{OC}
\end{equation}
If $x(t)$ is a solution of \eqref{OC}, and $u(t)$ is the associated control, we call $(x,u)$ an \emph{optimal pair}. We recall that $x$ is uniquely defined thanks to Carathéodory's theorem (see for instance, \cite{Carat}).
To solve such problem we recall the classical necessary condition for the optimality of a couple $(x,u)$ given by the Pontrjagin's Maximum Principle (PMP) \cite{Agrachev04}. Define the Hamiltonian of (\ref{OC}) as
$$H(x,\l,\l^0,u)=H_0(x,\l)+\sum_{i=1}^{m}u_iH_i(x,\l)+\l^0\varphi(x,u),$$ with
 $H_i(x,\l)=\langle \l,F_i(x)\rangle$, where $\langle\; ,\;\rangle$ denotes the dot product in $\R^n$, $i=0,\dots,m$, and $(\l,\l^0)\in\mathbb{R}^n\times\mathbb{R}$.
 The problem reduces to finding solutions of a Hamiltonian system in $\mathbb{R}^n\times\mathbb{R}^n$, a consequence of the Pontrjagin Maximum Principle (PMP): 
 	\begin{theorem}[PMP]
		If $(x,u)$ is an optimal pair, then there exists a Lipschitz curve $\lambda:[0,t_f]\rightarrow \mathbb{R}^n$ and a constant $\lambda^0 \leq 0$, such that $(\lambda,\lambda^0)\neq(0,0)$ and
		\begin{itemize}
		    \item[(i)]	$(x,\l)$ is a solution of  
		\begin{equation}
		\begin{cases}
		\dot{x}=\frac{\partial H}{\p \l}(x,\l,\l^0,u),\\
		\dot{\lambda}=-\frac{\p H}{\p x}(x,\l,\l^0,u).
		\end{cases}
		\label{ext}
		\end{equation}

		\item[(ii)] $H(x(t),\l(t),u(t))=\max_{\tilde{u}\in U}H(x(t),\l(t),\tilde{u})$ (without loss of generality, we stop writing the dependence on $\lambda^0$). 
		
		\item[(iii)] $H(x(t_f),\l(t_f),u(t_f))= 0$.
		\end{itemize}
		\label{pont}
	\end{theorem}
	
Thus, provided that
		\begin{equation}
		    H^{\max}(x,\l) = \max_{\tilde{u}\in U} H(x,\l,\tilde{u})
		    \label{Hmax}
		\end{equation} is $\mathcal{C}^2$-smooth, optimal solutions are just projections on $\R^n$ of the solutions of the Hamiltonian system defined by $H^{\max}$ given by the canonical projection $\pi:(x,\l)\in\R^n\times\R^n\mapsto x\in\R^n.$ Such pair $(x,\l)$ is called an \emph{extremal}, and its projection on $\R^n$ is an \emph{extremal trajectory.} 
  
In order to solve the two-boundary problem (\ref{OC}), we need to find the initial condition $\l_0=\l(0)$. To that end, we use a shooting method on the Hamiltonian $H^{\max}$. 
 
 Let us denote $z(t,z_0)\in\R^n\times\R^n$ the flow of the Hamiltonian system associated with $H^{\max}$. We will need the following definition.
\begin{definition}
The map $$\exp_{x_0}^t:\l_0\in \R^n\mapsto \pi(z(t,z_0))\in \R^n,$$ 
where $z_0=(x_0,\lambda_0) \in \mathbb{R}^n \times \mathbb{R}^n$ and $\pi$ is the projection of the variable $x$, is called the \emph{exponential map}.
\label{expo}
\end{definition}
We also define the \emph{shooting function} as 
\begin{equation}\label{shoot}
    Shoot(\l_0)=\exp_{x_0}^{t_f}(\l_0)-x_f.
\end{equation}
We will use nonlinear Newton-type methods to find zeros of $Shoot$ (see Section \ref{sec:Ocontrol}).

Finally, we remark that when $\l^0=0$ (\emph{abnormal} extremal) the Hamiltonian of the optimal solutions does not depend on the cost of the problem. For the \emph{normal} case, $\l^0\neq0$, the pair $(\l,\l^0)$ can be normalized as desired without loss of generality due to linearity in $\l$ (see, for instance \cite{Agrachev04}), so, in what follows, we will set $\l^0=-1/2$.

\section{Controllability of systems with a limit cycle}
\label{sec:control_LC}

Models in neuroscience, either of single cells or neural populations, exhibit, in general, oscillatory behavior, at least for some values of the parameters. It is often a challenging task to achieve global controllability for high-dimensional systems (dimension higher than 3) with nonlinear dynamics, particularly when the control is scalar. However, attaining local controllability around the periodic orbit turns out to be a more feasible objective. 
Next, we present a novel result that provides sufficient conditions for local controllability around a limit cycle. 
In the statement, we use the classical notation $ad_F(G)=[F,G]$, where $F$ and $G$ are vector fields and $[F,G]$ is the Lie bracket.
 \begin{proposition}
 \label{LocalCtrl}
 Let $\Gamma$ be a periodic orbit of period $T$ of the system defined by a vector field $F_0$ on $\mathbb{R}^n$. Assume that
 \begin{itemize}
    \item[(i)] $0$ is in the interior of the convex hull of $U$, where U is the control set;
    \item[(ii)] $\exists \, x\in\Gamma$ such that $\mathrm{rank}(\{\ad^{k}_{F_0}F_i(x)\}_{1\leq i\leq m, k\geq 0})=n$.    
\end{itemize}
Then, system (\ref{AC}) is controllable in a neighborhood of $\Gamma$, in time $t\geq T$. In particular,
    $\Gamma\subset int\mathcal{A}(x)$ for every $x\in\Gamma.$
 \end{proposition}

\begin{proof}
It comes as a consequence of Theorems \ref{lincont} and \ref{localcont}. We need to check that $(\ref{lincrit})$ holds for the linearized system of (\ref{AC}) along $\Gamma$.
 Let $\gamma(t)$ be a trajectory of $F_0$ corresponding to the periodic orbit $\Gamma$. The linearized system around $\gamma(t)$ writes as 
\begin{equation}
    \dot{x} = \underbrace{DF_0(\gamma(t))}_{A(t)}x+\underbrace{(F_1(\gamma(t),\dots,F_m(\gamma(t)))}_{B(t)}u,
\end{equation}
where $DF$ denotes the differential of the vector field $F$.
Then, 
\begin{eqnarray*}
B_1(t)&=&\dot{B}(t)-A(t)B(t) \\
&=&(DF_1(\gamma(t))\dot{\gamma}(t),\dots,DF_m(\gamma(t))\dot{\gamma}(t))
-DF_0(\gamma(t))(F_1(\gamma(t)),\dots,F_m(\gamma(t)))\\
&=&([F_0,F_1],\dots,[F_0,F_m])(\gamma(t)).
\end{eqnarray*}

By induction, we get $B_i(t)=(\textrm{ad}^i_{F_0}F_1,\dots,\textrm{ad}^i_{F_0}F_m))(\gamma(t)).$ The condition of Theorem \ref{lincont} is checked and the linearized system is controllable. By Theorem \ref{localcont}, this implies the local controllability around $\Gamma$ in one period. Furthermore, one can note that $\mathcal{A}(x)=\mathcal{A}(y)=\bigcup_{z\in\Gamma}\mathcal{A}(z)$ for all $x,y \in\Gamma$. Therefore, the proposition is proven.
\end{proof}

We investigate local controllability for two mean-field models that exhibit oscillatory behavior, to which we will later apply the control in the context of a problem of neuronal communication. Despite not being a property much explored in neuroscience models, local controllability holds in several classical models in Neuroscience; we provide some examples of it in the Appendix.

\paragraph{Exact mean-field models for neuronal populations.}
We consider an exact mean-field model \cite{Montbrio1, Montbrio2} describing the macroscopic dynamics of a population of inhibitory neurons in terms of the mean membrane potential $V$, the firing rate $r$, and the mean synaptic activation $S$, to which we add a control term to the mean voltage $V$ equation. Namely,
\begin{equation}
\begin{cases}
\tau_m\dot{r} = \dfrac{\Delta}{\pi\tau_m} + 2Vr,\\
\tau_m\dot{V} = V^2-(\tau_m\pi r)^2-\tau_m J S+I(t)+\tau_m\,u(t),\\
\tau_d\dot{S} = -S + r,
\end{cases}
\label{Inhib}
\end{equation}
where $\tau_m$ and $\tau_d$ are time constants modeling neural interactions, $J$ is the synaptic strength and $\Delta$ is a parameter controlling the heterogeneity of the cells in the network, associated to the width of a Lorentzian distribution (see \cite{Montbrio1} for more details). The term $I(t)$ refers to the external current; for the computations of this section we considered it to be constant $I(t) \equiv \bar{I}$. We will use the following set of parameter values for this system:
\begin{equation}\label{eq:Iparam}
\mathcal{P}_I=\{\Delta = 0.3,\; \tau_m = 10,\; \tau_d = 10,\; J= 21,\; \bar{I} = 4\}.
\end{equation}

Following the notation of system \eqref{AC}, we define
$$F_0(r,V,S)=\begin{pmatrix}
\left(\Delta/(\pi\tau_m) + 2Vr \right)/\tau_m\\
\left(V^2-(\tau_m\pi r)^2-\tau_m J S+I \right)/\tau_m\\
\left(-S+r \right)/\tau_d
\end{pmatrix} 
\quad \textrm{and}  \quad F_1=\begin{pmatrix} 0\\
1\\
0
\end{pmatrix}.$$

\begin{corollary}
If system \eqref{Inhib} has a periodic orbit $\Gamma$ for a specific set of parameters and $u\equiv 0$, then it is controllable around $\Gamma$.
\end{corollary}

\begin{proof}
We have $$ F_1 = \begin{pmatrix}0\\
1\\
0
\end{pmatrix}, \qquad
[F_0,F_1]=\frac{-1}{\tau_m}\begin{pmatrix} 2r\\
2V\\
0
\end{pmatrix},$$
and \[[F_0,[F_0,F_1]] = \frac{2}{\tau_m ^2}
\begin{pmatrix} 
2Vr - \Delta/(\pi\tau_m)\\
V^2-(\tau_m\pi r)^2+J\tau_m S-I\\
r\tau_m/(2\tau_d)
\end{pmatrix}.\]
Thus, along a non-trivial periodic orbit, the above vector fields ($F_1$, $ad_{F_0}F_1$ and $ad^2_{F_0}F_1$) generate the whole tangent space unless $r\equiv 0$ on the whole orbit. This would imply $\Delta=0$, which excludes the possibility of having oscillations, and so Proposition \ref{LocalCtrl} applies.
\end{proof}

We also study an exact mean-field model describing the macroscopic dynamics of two populations of neurons, one excitatory (E) and one inhibitory (I) \cite{Dumont19}, which follows the formalism developed in \cite{Montbrio1}. The controlled model consists of a set of differential equations for the E-population,
\begin{equation}
    \begin{cases}
     \tau_e\dot{r}_e = \Delta_e/(\pi \tau_e) + 2 r_e V_e,\\
     \tau_e\dot{V}_e = V_e^2 + \eta_e - (\tau_e\pi r_e)^2 + \tau_e S_{ee} - \tau_e S_{ei} + I_e(t) + \tau_e u(t),\\
     \tau_{si}\dot{S}_{ei} =  -S_{ei} + J_{ei}r_i,\\
     \tau_{se}\dot{S}_{ee} =  -S_{ee} + J_{ee}r_e,
     \end{cases}
     \label{Epop}
     \end{equation}
     and another identical set for the I-population,
    \begin{equation}
    \begin{cases}
     \tau_i\dot{r}_i = \Delta_i/(\pi\tau_i) + 2r_i V_i,\\    
    \tau_i\dot{V}_i = V_i^2 + \eta_i - (\tau_i
    \pi r_i)^2 + \tau_i S_{ie} - \tau_i S_{ii} + I_i(t) + \tau_i u(t),\\
   \tau_{se}\dot{S}_{ie} = -S_{ie} + J_{ie} r_e,\\
   \tau_{si}\dot{S}_{ii} = -S_{ii} + J_{ii}r_i.
    \end{cases}
    \label{Ipop}
\end{equation}
Similar to system \eqref{Inhib}, $r_k$ and $V_k$ ($k \in \{e,i\}$) represent the firing rate and the mean voltage for each population $k$. The variable $S_{ab}$ models the synaptic interaction from population $b$ to population $a$. The terms $I_k(t)$ ($k \in \{e,i\}$) refer to the external current applied to population $k$.
Here, we consider
\[I_{k}(t) \equiv \bar{I}_k,\]
where $\bar I_k$ is a tonic current. In Section \ref{sec:CTC} the external current will be periodic.

In what follows we will set the parameters $J_{ee}=J_{ii}=0$. Thus, the dynamics reduces to a 6 dimensional system where $S_{ee} = S_{ii} = 0$. Along the paper, the values of the other parameters will be 
\begin{equation}\label{eq:EIparam}
\begin{array}{rc}
\mathcal{P}_{EI} = & \{ \Delta_e = 1,\; \Delta_i = \Delta_e,\; \eta_e = -5,\; \eta_i = \eta_e,\;
\tau_e = 10,\; \tau_i = \tau_e,\; \tau_{si} = 1,\; \tau_{se} = 1,\; \\
& J_{ei} = 15, \; J_{ie} = J_{ei},\; \bar{I}_e = 10,\; \bar{I}_i = 0 \}.
\end{array}
\end{equation}

As for system \eqref{Inhib}, we apply Proposition \ref{LocalCtrl} to show that the dynamics of the E-I population system \eqref{Epop}-\eqref{Ipop} is controllable around its periodic orbit. In this case, we need to rely on numerical computations to validate the hypothesis of Proposition \ref{LocalCtrl}. Thus, let us define $D_{EI}(t)$ as the determinant of the matrix
\begin{equation}
\label{matrixcont}
A(t) = \mbox{col}(A_i)(\gamma(t)),\quad \mbox{with } A_1=F_1 \mbox{ and } A_i=ad^{i-1}_{F_0}F_1, i=2,\dots, 6,
\end{equation}
where $F_0$ is the vector field defining system \eqref{Epop}-\eqref{Ipop} with parameter values $\mathcal{P}_{EI}$ and $u\equiv 0$, $\gamma (t)$ is the trajectory corresponding to the periodic orbit of the system and $F_1=(0,1,0,0,1,0)^T$.
 In Figure \ref{fig:EIcontrollability}(a) we show the determinant $D_{EI}$ (indeed the logarithm of the determinant) and we can observe that it is clearly non-zero for the points of the limit cycle close to the peak of $S_{ei}$ (see Figure  \ref{fig:EIcontrollability}(b)). 

\begin{figure}
\begin{tabular}{ll}
(a) & (b) \\
 \includegraphics[scale=0.3]{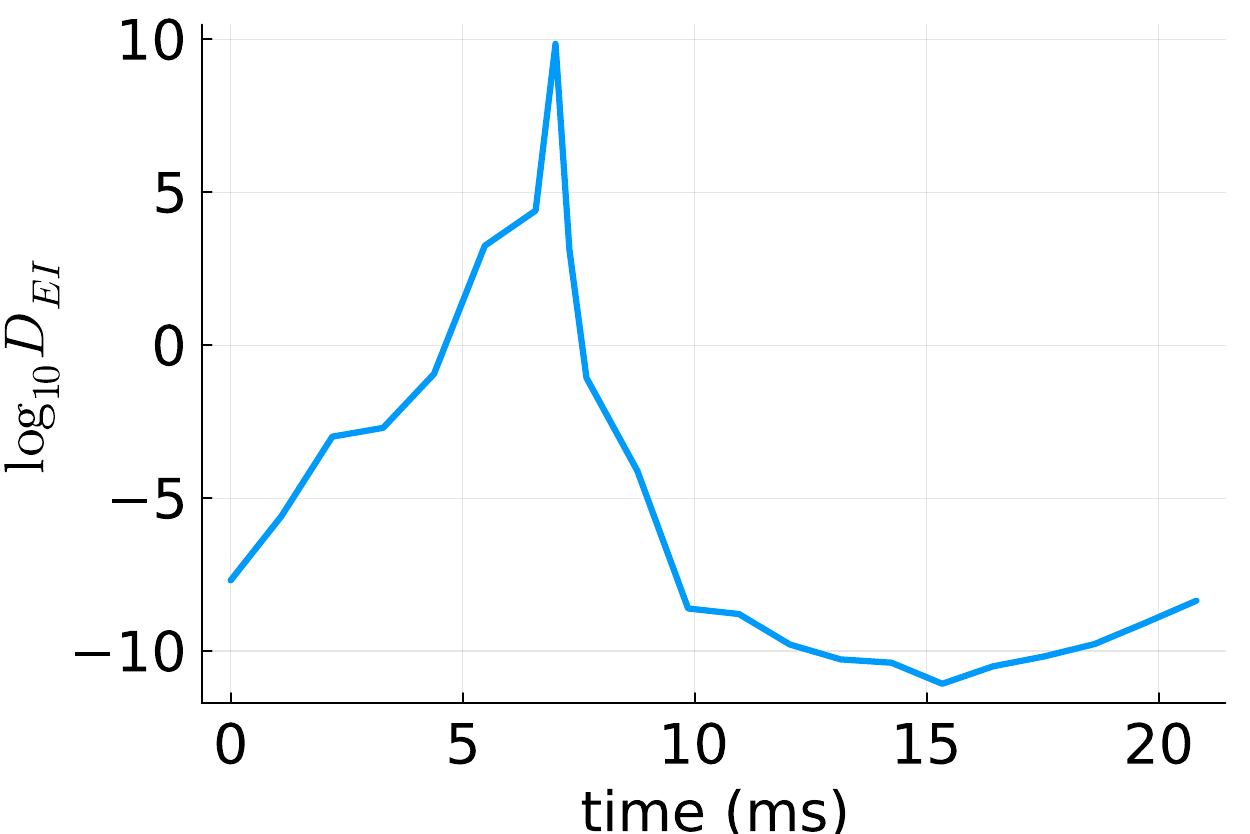}
  &
  \includegraphics[scale=.3]{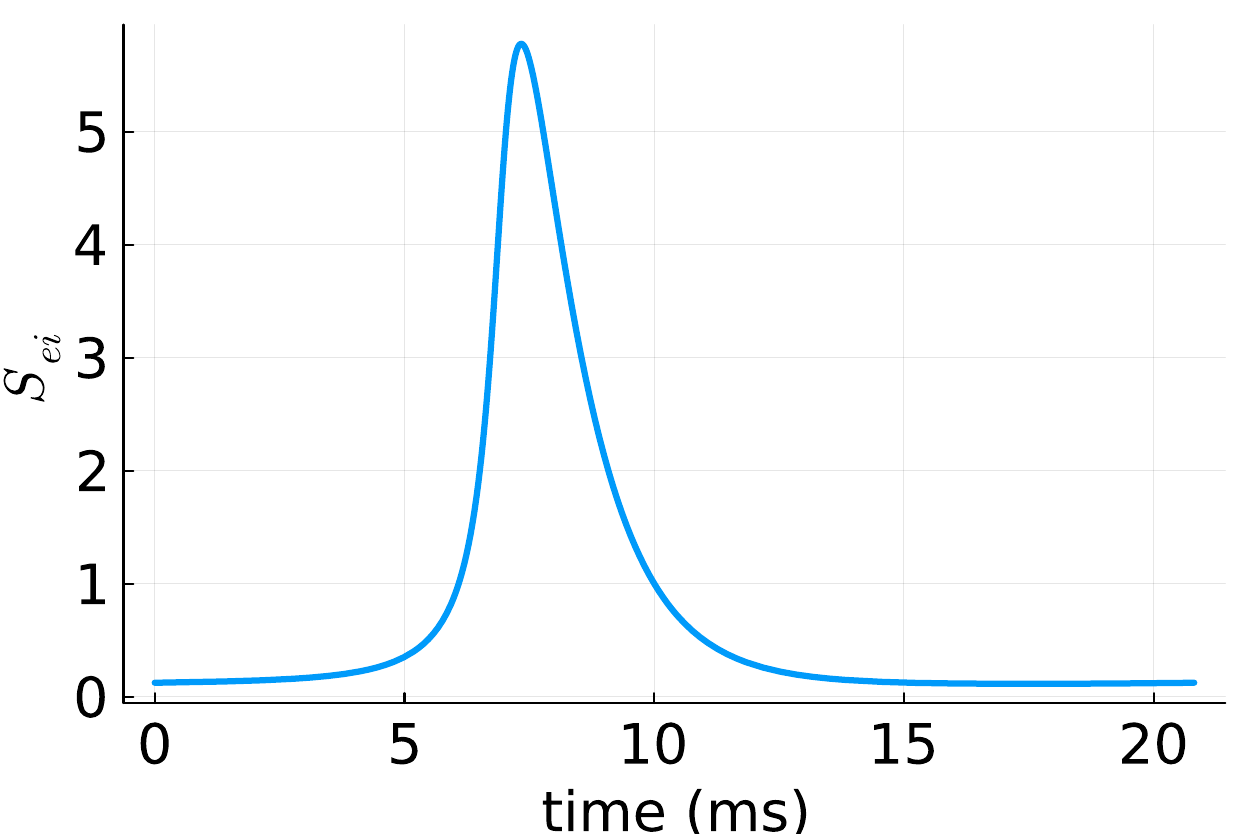}
\end{tabular}  
  \caption{(a) Logarithm of the determinant $D_{EI}$ of the matrix \eqref{matrixcont} and (b) synaptic coordinate $S_{ei}$ along the periodic orbit $\Gamma$ of the E-I network model \eqref{Epop}-\eqref{Ipop}  with $u\equiv 0$ and the set of parameters defined in \eqref{eq:EIparam}.
  }
\label{fig:EIcontrollability}
\end{figure}

\section{The phase-amplitude reduction}\label{sec:PAreduc}

Our control strategy, inspired by \cite{Moehlis}, uses extensively the phase-amplitude reduction of a dynamical system around a stable limit cycle \cite{GH09, CGH13,Perez_Cervera20}. We recall the principle of this reduction in this section.

Let $F:\mathbb{R}^n\rightarrow\mathbb{R}^n$ be an analytic vector field and
\begin{equation}
    \dot{x} = F(x),
    \label{dyn}
\end{equation}
 a dynamical system with a stable hyperbolic limit cycle $\Gamma$ parameterized by the phase
 \begin{equation}\label{eq:lcgamma}
 \gamma:\theta\in\mathbb{T}:= \mathbb{R} \backslash \mathbb{Z} \mapsto \gamma(\te)\in\R^n.  
 \end{equation}
Let us denote by $\mathcal{B}$ the basin of attraction of the limit cycle $\Gamma$. By the stable manifold theorem \cite{HPS}, we can extend the phase definition to the whole basin of attraction of the limit cycle. Indeed,  $\forall \, y_0 \in \mathcal{B}, \exists \, \te \in \mathbb{T}$ such that $|x(t, y_0)-\gamma(t/T+\te)| \underset{t\rightarrow\infty}{\longrightarrow}0$. Here, $x(t,y_0)$ is the flow of the vector field \eqref{dyn}. Therefore, we can define a function $\Theta$ on $\mathcal{B}$ such that $\Theta(y_0)=\te$ (see \cite{Guc}). 
The set of points with the same phase $\theta$, $\mathcal{I}(\theta)=\{x\in \mathcal{B} \, | \, \Theta(x) = \theta \}$ is called the $\theta$-\emph{isochron}. The isochrons are the leaves of the stable manifold of the limit cycle and the flow at time $t$ sends $\mathcal{I}(\theta)$ to $\mathcal{I}(\theta+t/T)$.

Assuming certain conditions on the Floquet exponents of the limit cycle $\Gamma$, one can prove (see \cite{CFL05, CLMJ15, Perez_Cervera20}) that there exists an analytic diffeomorphism 
\begin{equation}\label{Kparam}
K:(\te,\s)\in {\mathbb T}\times {\mathbb R}^{n-1}\mapsto x\in\mathbb{R}^n, 
\end{equation}
such that system \eqref{dyn} writes as
$$\begin{cases}
\dot{\theta} = \dfrac{1}{T},\\
\dot{\sigma} = \Lambda\sigma,
\end{cases}$$
in the $(\te,\s)$ coordinates, with $\Lambda=\textrm{diag}(\mu_i)_{1\leq i\leq n-1},$ being the diagonal matrix of the Floquet exponents of the periodic orbit $\mu_i \in \mathbb{R}$. Here we assume that the Floquet exponents are real and distinct, more precisely, $\mu_{n-1}<\cdots< \mu_{1}<0 .$ Thus, we have that
\[x(t,K(\theta_0,\sigma_0))=K(\theta_0+t/T, \sigma_0 e^{\Lambda t}),\]
where $x(t,x_0)$ is the flow of the vector field \eqref{dyn}, with $x_0=K(\theta_0,\sigma_0)$. 

\begin{remark}\label{rem:floquet}
The assumption for distinct eigenvalues is to ensure that there are no resonances and the system can be transformed into a linear system in $\sigma$. However, it is not necessary to have real eigenvalues. Indeed, the case of complex eigenvalues is similar (see for instance the discussion in \cite{CLMJ15}). For the purposes of this paper, we only need that the Floquet exponent with smallest modulus is real, as we will see in the example considered later on.  
\end{remark}

The variables $\sigma \in \mathbb{R}^{n-1}$ are typically referred to as the \emph{amplitude coordinates} \cite{GH09,CGH13,Perez_Cervera20} and provide a measure to quantify the proximity to the limit cycle. Analogously to the $\Theta$ function, we can define a vector-valued function $\Sigma$ on $\mathcal{B}$ such that $\Sigma(y_0)=\sigma \in \mathbb{R}^{n-1}$. The set of points with the same amplitude, $\mathcal{J}(\sigma)=\{x\in \mathcal{B} \, | \, \Sigma(x) = \sigma \}$ is called the $\sigma$-\emph{isostable}.
Notice that the vector-valued function $(\Theta,\Sigma)$ is the inverse of $K$, that is, $K\circ(\Theta,\Sigma)(x)=x.$ 

In this framework, when perturbing a system, two functions are of importance: the \emph{Phase Response Function} (PRF) and the \emph{Amplitude Response Function} (ARF), which measure changes in the phase and amplitude of an oscillation, respectively, induced by a perturbation $\Delta x$ as a function of the point at which it is received. Mathematically,
\[PRF(x)=\Theta(x+\Delta x) - \Theta(x), \]
and
\[ARF(x)=\Sigma(x + \Delta x) - \Sigma(x).\]

When applying a perturbation $p(t)$ (not necessarily small) to system \eqref{dyn}, that is, $\dot{x} = F(x)+ p(t)$, the evolution of the $(\theta,\sigma)$ variables is given by the following perturbed system:
\begin{equation}\label{eq:PA}
\begin{cases}
\dot{\theta} = \dfrac{1}{T}+\nabla\Theta(K(\theta,\sigma)) \cdot p(t),\\
\dot{\sigma} = \Lambda\sigma+\nabla\Sigma(K(\theta,\sigma)) \cdot p(t).
\end{cases}
\end{equation}
The functions $\nabla\Theta(K(\theta,\sigma))$ and $\nabla\Sigma(K(\theta,\sigma))$ correspond to the first order approximation of the PRF and the ARFs and are called the \emph{infinitesimal phase and amplitude response functions}, respectively, i.e. $iPRF(\theta,\sigma)=\nabla\Theta(K(\theta,\sigma))$ and $iARF(\theta,\sigma)=\nabla\Sigma(K(\theta,\sigma))$. Computing the values of the iPRF and the iARF (as well as the parameterization $K$ in \eqref{Kparam}) globally requires efficient numerical algorithms (see \cite{Perez_Cervera20} for efficient numerical methods). To circumvent the expense of such numerical computations, many studies rely on the \emph{weak coupling approximation}: if the perturbation is small, the resulting trajectory stays close to the limit cycle, and thus  $iPRF(\te,\s) \approx iPRF(\te,0) =: Z_0(\te)$,  $iARF(\te,\s) \approx iARF(\te,0) =: I_0(\te)$ (these functions are called \emph{infinitesimal phase and amplitude response curves}, respectively). Though analytic computations of these curves are out of reach, except for very simple cases, one can easily compute them numerically using the fact that the functions $Z_0$ and $I_0=(I_{0,1},\ldots, I_{0,n-1})$ are periodic solutions of the following linear differential equations (see \cite{CGH13,GH09,EK91}):
\begin{equation}\label{eq:eqPRC}
\dfrac{1}{T} \dfrac{d }{d \theta} Z_0(\te) = -DF^T (\gamma(\te))\, Z_0(\te),
\end{equation}
and,
\begin{equation}\label{eq:eqARC}
\dfrac{1}{T} \dfrac{d }{d \theta} I_{0,i} (\te) = (\mu_i \textrm{Id} - DF(\gamma(\te)))\, I_{0,i}(\te),
\end{equation}
with a normalisation condition. Recall that $\gamma$ is the parameterization of the limit cycle given in \eqref{eq:lcgamma}. In Figure \ref{fig:PRCARC}, we show the iPRC and the iARC for the least contractive normal direction (associated to the largest Floquet exponent) for systems (\ref{Inhib}) and \eqref{Epop}-\eqref{Ipop}.
The PRC is a very useful tool for the study of oscillators and of primary importance for biologists as it can be measured experimentally, see \cite{ErmentroutTerman10,SPB11} for more details and a complete study of the PRC.

In this paper, we will apply a scalar control in a given direction ${\bf v} \in \mathbb{R}^n$, that is,
\begin{equation}\label{pertAP}
\dot{x} = F(x)+ u(t)\, {\bf v}, 
\end{equation}
and we will study the control system in terms of the phase-amplitude variables 
\begin{equation}\label{pertAPred}
\begin{cases}
   \dot{\te}=\dfrac{1}{T}+ u(t)\, \nabla \Theta (K(\theta,\sigma)) \cdot {\bf v},\\
   \dot{\s}=\Lambda\s + u(t)\, \nabla \Sigma (K(\theta,\sigma)) \cdot {\bf v}.
\end{cases}
\end{equation}

By \emph{controlled trajectory}, we will refer to a solution of (\ref{pertAP}). An \emph{original trajectory} will be a solution of (\ref{dyn}). 

We will both work within and beyond the weak coupling approximation. That is, we will consider approximations of the  functions $\nabla \Theta$ and $\nabla \Sigma$ in \eqref{pertAPred} by the iPRC $Z_0(\theta)$ and iARC $I_0(\theta)$, respectively (weak coupling hypothesis), and we will also include its first order terms in $\sigma$, to be able to treat the case of a larger control while keeping precision in the numerical resolution of our dynamics. We denote respectively $Z_1$ and $I_1$ the first order terms of 
$\nabla\Theta$ and $\nabla\Sigma$ in $\s$: 
\begin{equation}\label{eq:nablatheta}
\nabla\Theta(K(\theta,\sigma)) = Z_0(\te) +  Z_1(\te)\s + O(\sigma^2),
\end{equation}
and  
\begin{equation}\label{eq:nablasigma}
\nabla\Sigma(K(\theta,\sigma)) = I_0(\te) + I_1(\te)\s + O(\sigma^2).
\end{equation}
These terms can be computed from the change of coordinates $K$.

Moreover, we will consider only the largest (the smallest in absolute value) Floquet exponent $\mu := \mu_1$, and (abusing the notation) the associated $\sigma = \s_1$ coordinate (we will see in the next section that for the examples considered there are several orders of magnitude between the first and the second Floquet multiplier). 

\begin{figure}
\begin{tabular}{ll}
    (a) & (b) \\
    \includegraphics[scale=0.3]{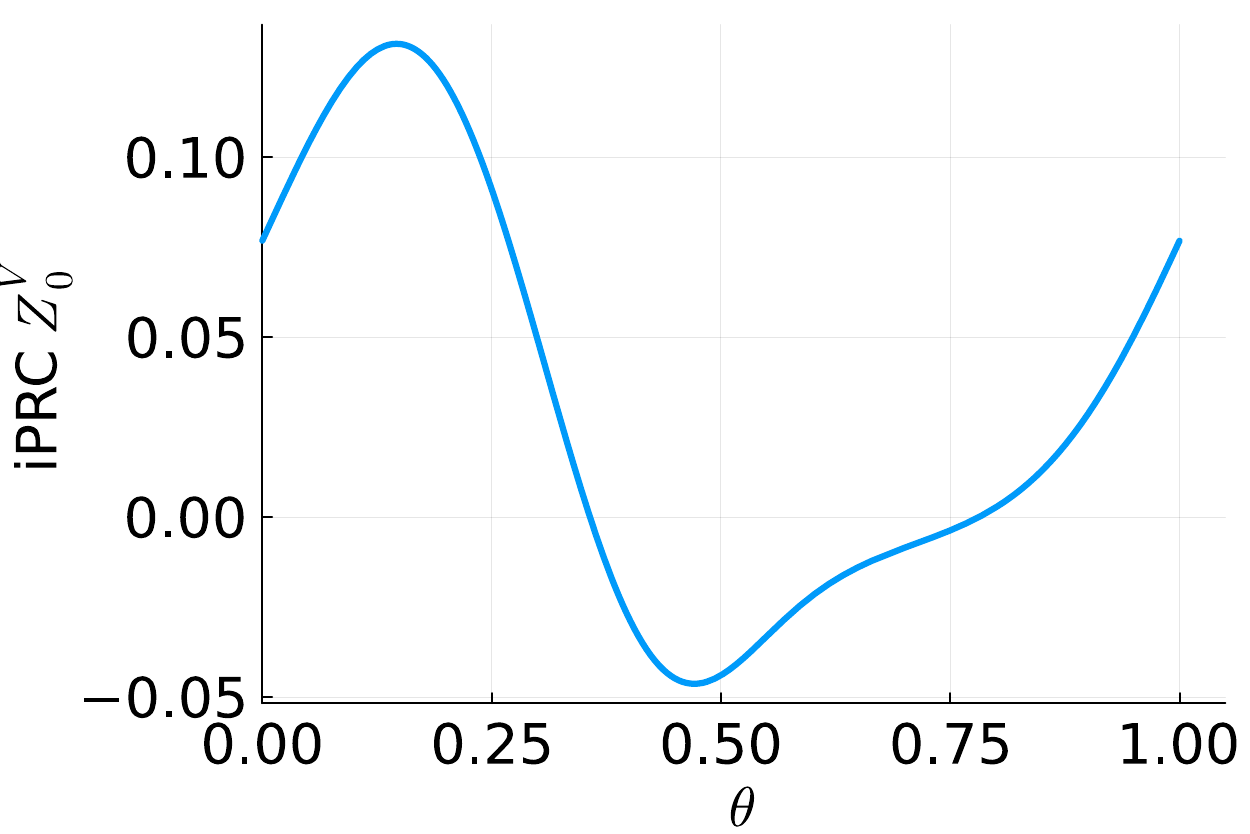} &
    \includegraphics[scale=0.3]{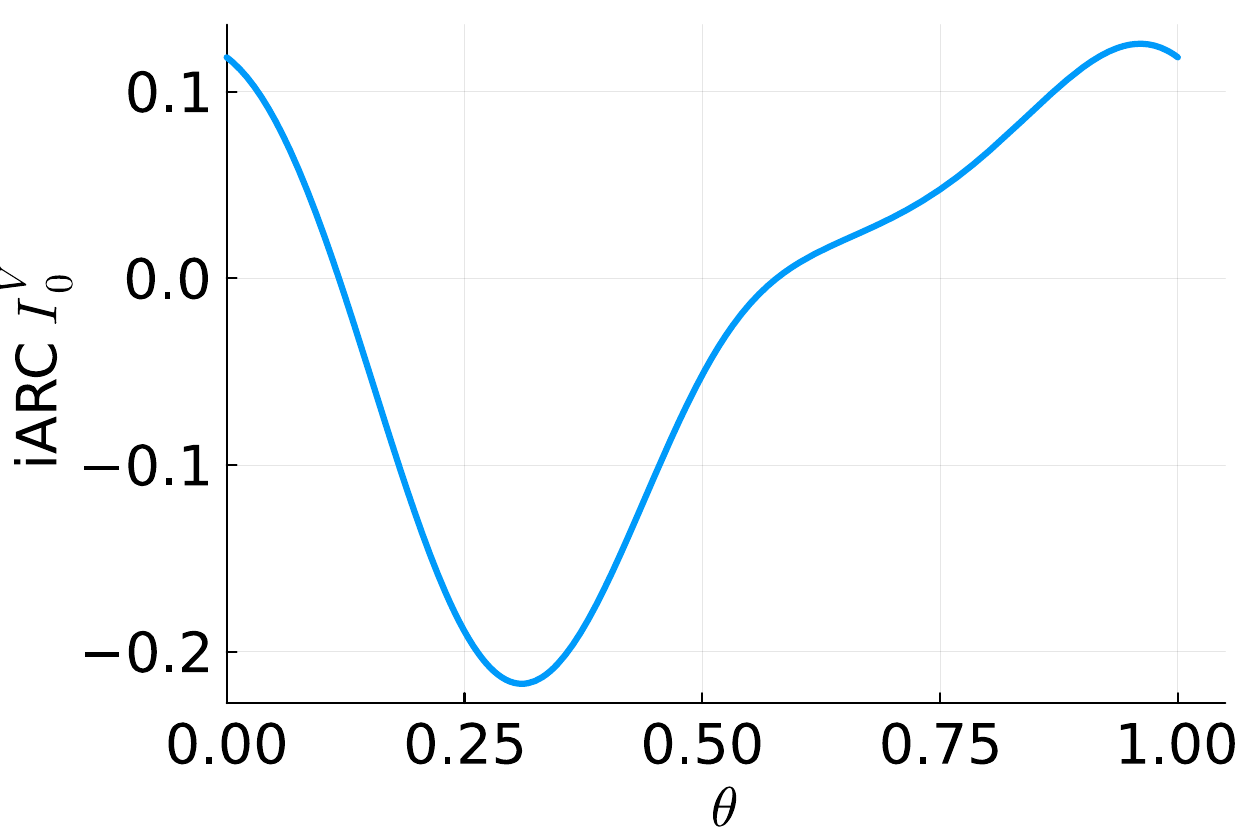} \\ 
    (c) & (d) \\
    \includegraphics[scale = 0.3]{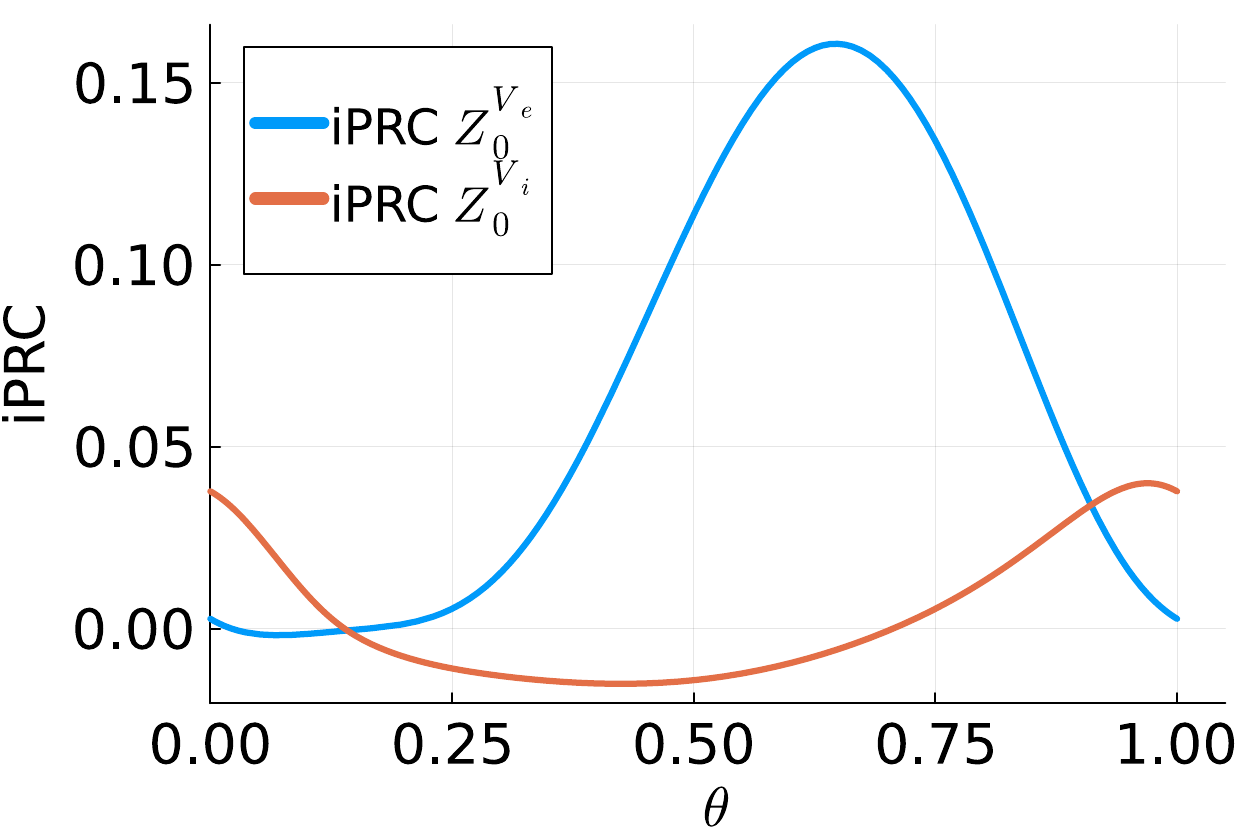} &
    \includegraphics[scale = 0.3 ]{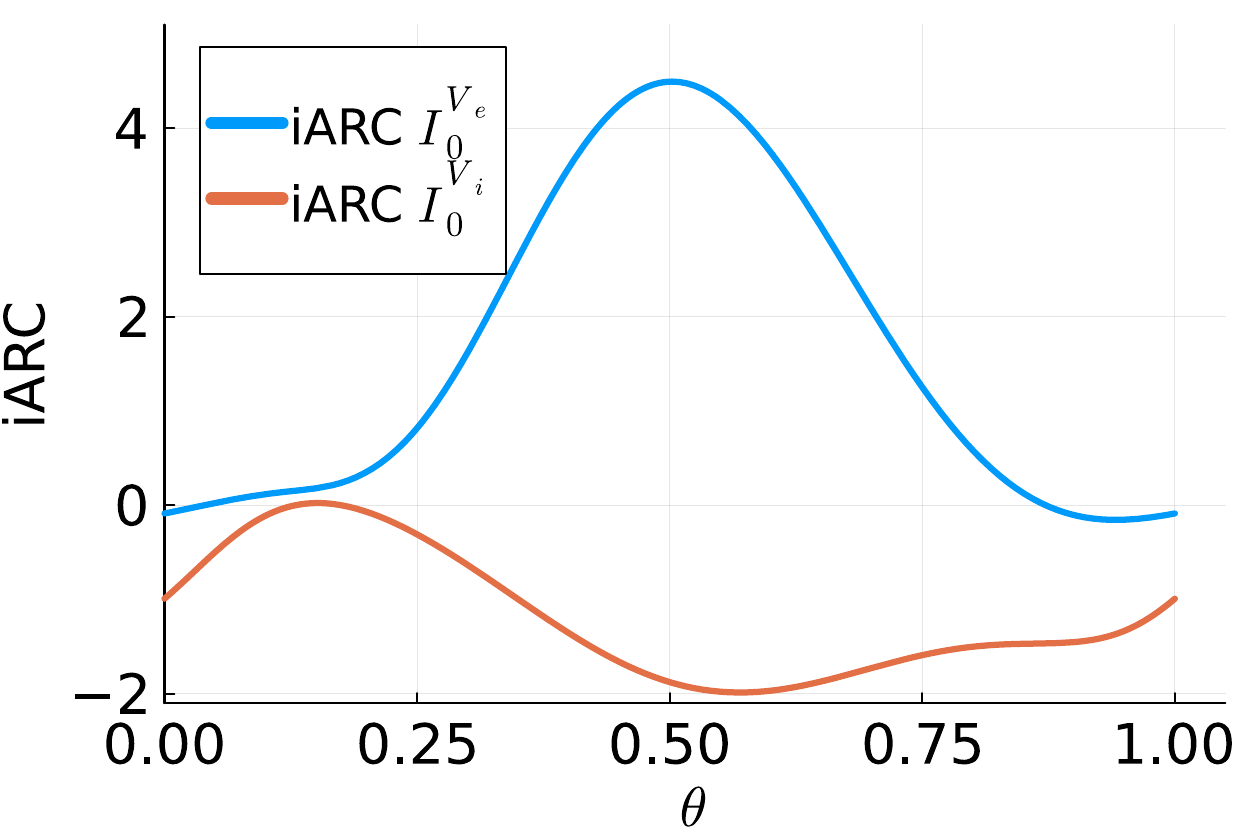} \\
\end{tabular}    
\caption{(a) V component of the iPRC ($Z_0^V$) and (b) iARC ($I_0^V$) of the least contractive normal direction for the limit cycle of system \eqref{Inhib} with $u\equiv 0$ corresponding to a self-inhibitory population with parameter values given in \eqref{eq:Iparam}; (c)  $V_e$ component (blue curve) and $V_i$ component (red curve) of the iPRC ($Z_0^{V_e}$ and $Z_0^{V_i}$) and (d) iARC ($I_0^{V_e}$ and $I_0^{V_i}$) of the least contractive normal direction for the limit cycle of system \eqref{Epop}-\eqref{Ipop} with $u\equiv 0$ corresponding to an E-I network with parameter values given in \eqref{eq:EIparam}.
}
\label{fig:PRCARC}
\end{figure}

\section{Optimal control using phase-amplitude variables}\label{sec:Ocontrol}

Our goal is to control the phase of an oscillating neural population (either the self-inhibitory network \eqref{Inhib} or the E-I network \eqref{Epop}-\eqref{Ipop}) by means of an external input to the population from a different neural source. We will assume that the brain is working at an optimal-energy regime, meaning that this external input, represented by the control, satisfies a minimum-energy hypothesis \cite{Friston2006, Friston2010, Gu2017}. To achieve this goal, phase-amplitude variables are more suitable for determining the optimal control, since they enable a more direct and targeted control strategy.

In this section, we provide an overview of control problems involving phase and amplitude variables and the numerical methods used to solve them. These control problems are designed to close the cycle within a specified time $t_f$ by applying an external input, the control, to the mean voltage equations. More precisely, the first problem uses the phase reduction (section \ref{sec:phase_control}) and the second one also includes the dynamics of the amplitude coordinate parameterizing the slow manifold (section \ref{sec:PAcontrol}).

In Section \ref{sec:CTC}, we will explore the applications to neural communication, particularly when the control is periodically applied over time.

\subsection{Phase-only minimum energy control}\label{sec:phase_control}

We present the problem of controlling the phase of a limit cycle using the phase reduction approach (assuming that the trajectory remains close to the limit cycle). The control problem writes as
\begin{equation}
\begin{cases}
\tag{$\textrm{OC}_1$}
\dot{\te}=\dfrac{1}{T}+ Z_0^{\bf v} (\te)u(t),\\
\te(0)=\te_0,\\
\te(t_f)=1,\\ 
\int_0^{t_f}u^2\rightarrow \min,
\end{cases}
\label{OCphase}
\end{equation}
where $Z_0^{\bf v}(\te):=Z_0(\te) \cdot {\bf v}$ and $Z_0$ is the iPRC given in \eqref{eq:nablatheta}. For the examples of this paper the direction $\bf v$ will be the voltage direction, namely, $Z_0^{\bf v}= Z_0^V$ for system \eqref{Inhib} and  $Z_0^{ \bf v}=Z_0^{V_e}+Z_0^{V_i}$ for the E-I system \eqref{Epop}-\eqref{Ipop}.

According to the PMP, a solution $\te(t)$ of such problem is the projection on the phase space of the solutions $(\te(t),\l_\te(t))$ of the Hamiltonian 
\begin{equation}
   H(\te,\l_\te,u) = \l_\te/T+\l_\te Z_0^{\bf v}(\te)u-u^2/2,
\end{equation}
together with the maximization condition $H(\te(t),\l_\te(t),u(t))=\max_{\tilde{u}\in \mathbb{R}}H(\te(t),\l_\te(t),\tilde{u})$ for all $t\in[0,t_f]$, which is achieved for 
\begin{equation} \label{eq:uoptimal1}
u = \l_\te Z_0^{\bf v}(\te).
\end{equation}
 
Thus, we get 
 $$H^{\max}(\te, \l_\te) =\max_{\tilde u \in \mathbb{R}} H (\te, \l_\te,\tilde u) = \frac{\l_\te}{T}+ \frac{(\l_\te Z_0^{\bf v}(\te))^2}{2}.$$
 The equations of motion for the Hamiltonian system are 
 \begin{equation}
    \begin{cases}
    \dot{\te}=1/T+\l_\te \,\left(Z_0^{\bf v}(\te)\right)^2, \\
    \dot{\l}_\te=-\l_\te^2\, Z_0^{\bf v}(\te)\,{Z_0^{\bf v}}'(\te).\\
    \end{cases}
    \label{HamP}
\end{equation}
To solve this problem we use the shooting algorithm presented in Section \ref{sec:method} in order to find the desired $\l_\te(0)$.

\begin{figure}
\begin{tabular}{ll}
(a) & (b) \\
\includegraphics[scale=0.3]{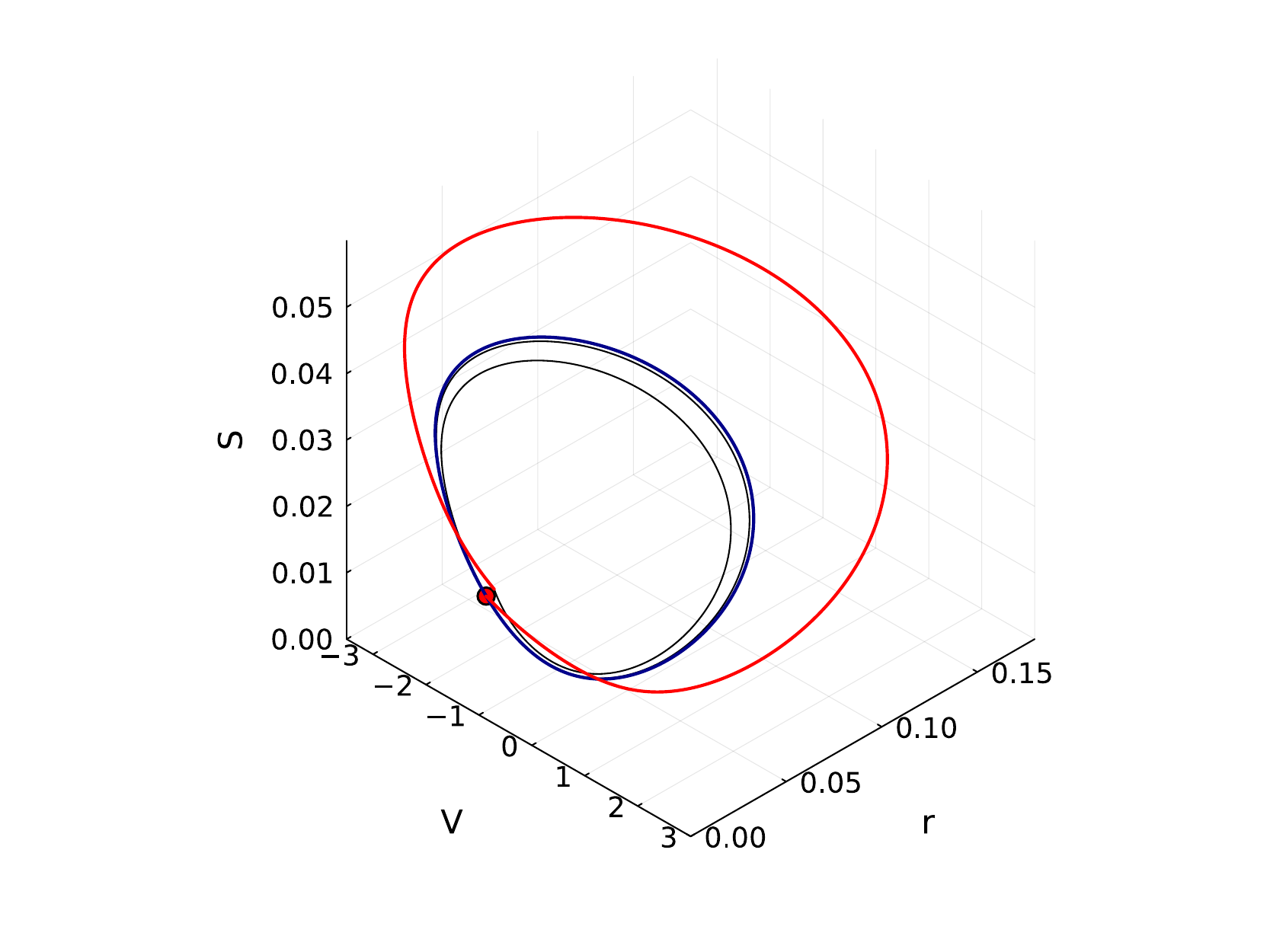} &
\includegraphics[scale=0.3]{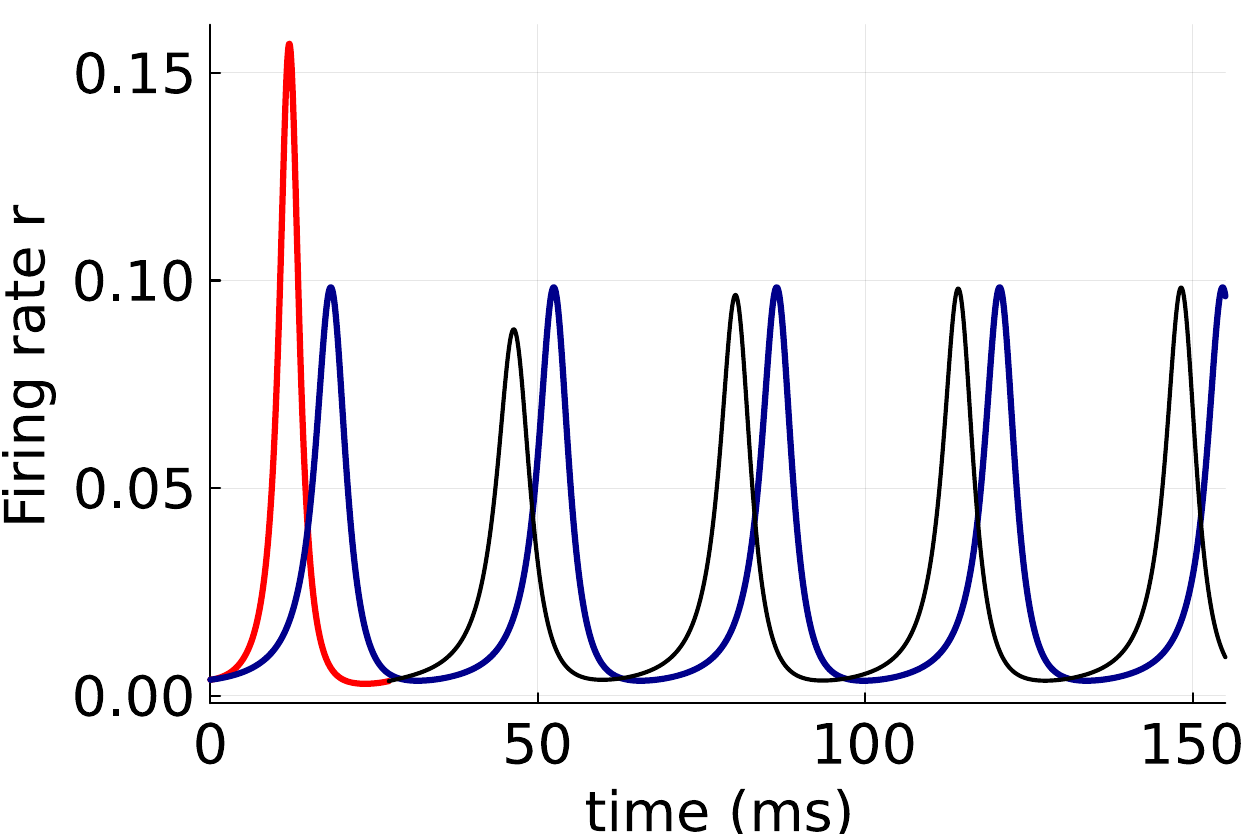} \\
(c) & (d) \\
\includegraphics[scale = 0.3]{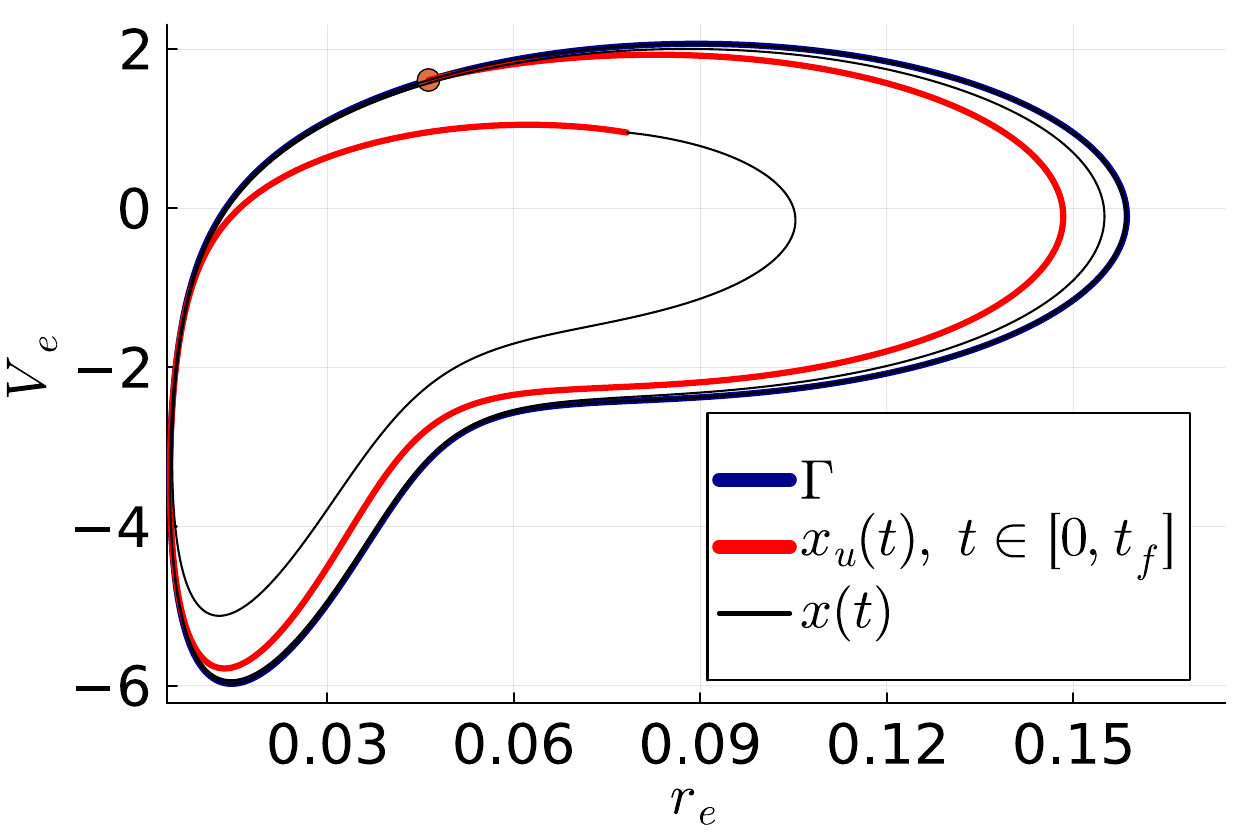} &
\includegraphics[scale = 0.3]{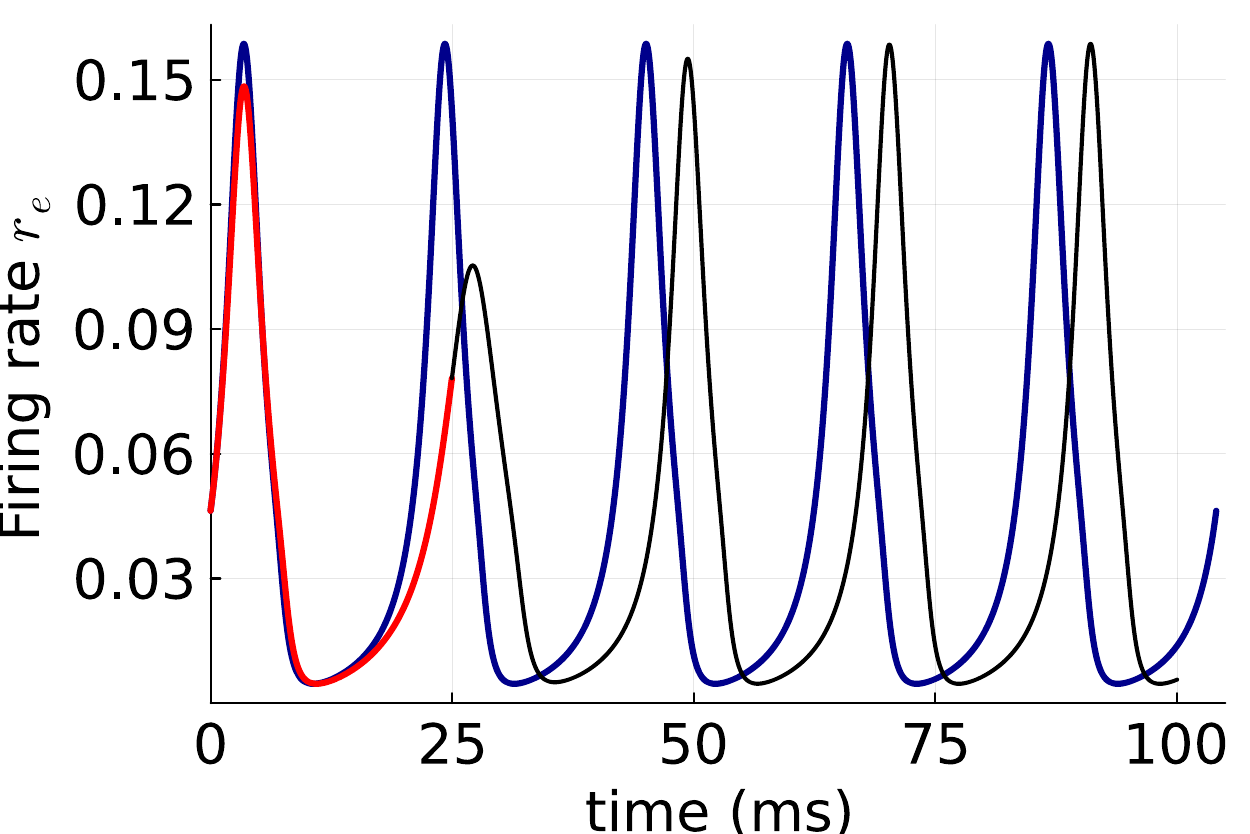}
\\

\end{tabular}
        \caption{ (a,b) Inhibitory population model \eqref{Inhib}. (a) Original ($u\equiv 0$) limit cycle of period $T\approx34.047$ (blue curve), controlled trajectory with the optimal control $u$, starting at the dot symbol, obtained by means of solving \eqref{OCphase} with $t_f=0.8T$ (red curve) and its continuation (i.e., for $t>t_f$) without control (black curve). (b) Time evolution of the firing rate variable $r$ for the trajectories in panel (a). (c,d) E-I population model \eqref{Epop}-\eqref{Ipop}. (c) Projection on the $(V_e,r_e)$ plane of the original ($u\equiv 0$) limit cycle of period $T\approx20.811$ (blue curve), the controlled trajectory (red curve) with the optimal control $u$  starting at the dot symbol, obtained by means of solving \eqref{OCphase} with $t_f = 1.2T$, and its continuation for $t>t_f$ without control (black curve).  (d) Time evolution of the firing rate variable $r_e$ for the trajectories in panel (c). 
        }
        \label{fig:conttraj}
\end{figure}

In Figure \ref{fig:conttraj}, we show the trajectory of system \eqref{Inhib} with $t_f=0.8T$ (panels (a) and (b)) and system \eqref{Epop}-\eqref{Ipop} with $t_f=1.2T$ (panels (c) and (d)) when applying a control $u$ obtained by solving the control problem \eqref{OCphase} using the PMP. Notice that, in both examples, the controlled system presents an orbit (red curve) that is displaced away from the original ($u\equiv 0$) limit cycle (blue) (see Figure \ref{fig:conttraj}(a) and (c)).  
This situation can be problematic for several reasons: leaving the basin of attraction, leaving the neighborhood of controllability, but, most importantly, breaking the weak coupling approximation and leading to an imprecise result. To overcome this problem, in Section \ref{sec:PAcontrol} we include the amplitude variable in the control problem. 

Notice that, once the control is turned off, the trajectory returns to the original limit cycle with a phase shift close to $0.2T$ (resp. $-0.2T$) when $t_f=0.8T$ (resp. $t_f=1.2T$), showing that the control is functioning as intended (see Figures \ref{fig:conttraj}c and d). 

\subsection{Phase-Amplitude control} \label{sec:PAcontrol}

An option to gain more accuracy and to avoid trajectories leaving the basin of attraction or the controllability region, especially when setting more drastic conditions (for instance, when $t_f$ is far away from $T$), is to penalize the distance to the limit cycle in the cost function. As in \cite{Moehlis}, this can be achieved by modifying the cost functional as 
\begin{equation}\label{eq:cost2}
C_{2}(x,u) = \int_0^{t_f} u(t)^2dt +\alpha\int_0^{t_f}\sigma(t)^2dt.
\end{equation}

Abusing of notation, we denote by $\sigma\in \R$ the amplitude coordinate in the direction of the largest Floquet exponent (recall that originally $\s\in\R^{n-1}$), and denote by $\mu$ the corresponding Floquet exponent. The parameter $\alpha$ quantifies the importance of the average squared distance, measured as the $L^2$ norm of $\sigma$, in the cost functional.

As a result of including the amplitude penalization, the optimal-control problem incorporates now the dynamics on the normal direction to the limit cycle, $\sigma$, and has the form: 
\begin{equation}
\begin{cases}
    \dot{\te}=\dfrac{1}{T}+  Z^{\bf v}(\te,\s)u(t),\\
     \dot{\s}=\mu\s + I^{\bf v}(\te,\s)u(t),\\
     \te(0)=0,\; \s(0)=0,\\
     \te(t_f)=1,\; \s(t_f)\textrm{ free},\\
   \int_0^{t_f} ( u^2+\alpha\sigma^2)dt\rightarrow \min,
    \end{cases}
    \label{OCP2}
    \tag{$OC_2$}
\end{equation}
where $Z^{\bf v}(\te,\s)= \nabla \Theta (K(\te,\s)) \cdot {\bf v}$ and $I^{\bf v}(\te,\s)= \nabla \Sigma (K(\te,\s)) \cdot {\bf v}$.  For the examples of this paper, the direction $\bf v$ will be the voltage direction, namely, $Z^{\bf v}= Z^V$, $I^{\bf v}= I^V$ for system \eqref{Inhib} and  $Z^{\bf v}=Z^{V_e}+Z^{V_i}$ and $I^{\bf v}=I^{V_e}+I^{V_i}$ for the E-I system \eqref{Epop}-\eqref{Ipop}.

The functions $\nabla \Theta (K (\te,\s))$ and $\nabla \Sigma (K(\te,\s))$ will be approximated using Taylor expansions in $\s$ given in \eqref{eq:nablatheta} and \eqref{eq:nablasigma}, respectively. Thus, $Z^{\bf v}(\te,\sigma)=Z_0^{\bf v}(\te)+\sigma Z_1^{\bf v}(\te) + \mathcal{O}(\sigma^2)$ and $I^{\bf v}(\te,\sigma)=I_0^{\bf v}(\te)+\sigma I_1^{\bf v}(\te) + \mathcal{O}(\sigma^2)$, where $Z^{\bf v}_i=Z_i \cdot {\bf v}$ and $I^{\bf v}_i=I_i \cdot {\bf v}$, for $i \geq 0$. In practical implementations, we will consider only the first dominant terms in $\sigma$.

Therefore, according to the PMP, the maximized Hamiltonian \eqref{Hmax} to solve the control problem \eqref{OCP2} is given by
$$H^{\max}(\te,\s,\l_\te,\l_\s)= \underbrace{\l_\te/T+\mu \l_\s\s}_{=H_0} + \frac{(\l_\te Z^{\bf v}(\te,\s)+\l_\sigma I^{\bf v}(\te,\s))^2}{2} - \frac{\alpha}{2}\s^2,$$
which is achieved for 
\begin{equation} \label{eq:uoptimal2}
u=\l_\te Z^{\bf v}(\te,\s)+\l_\sigma I^{\bf v}(\te,\s). 
\end{equation} 
Thus, the equations of motion for the  Hamiltonian system are
\begin{equation}\label{HamPA}
    \begin{cases}
    \dot{\te}=1/T+Z^{\bf v}(\te,\s)(\l_\te Z^{\bf v}(\te,\s)+\l_\sigma I^{\bf v}(\te,\s)),\\
    \dot{\s}=\mu\s+I^{\bf v}(\te,\s)(\l_\te Z^{\bf v}(\te,\s)+\l_\sigma I^{\bf v}(\te,\s)),\\
    \dot{\l}_\te=-(\l_\te Z^{\bf v}(\te,\s)+\l_\sigma I^{\bf v}(\te,\s))(\l_\te \partial_{\theta}Z^{\bf v}(\te,\s)+\l_\sigma \partial_{\te}{I^{\bf v}}(\te,\s)),\\
    \dot{\l}_\s=-\mu \l_\s+\alpha \s -(\l_\te Z^{\bf v}(\te,\s)+\l_\sigma I^{\bf v}(\te,\s))(\l_\te \partial_{\s}Z^{\bf v}(\te,\s)+\l_\sigma \partial_{\s}I^{\bf v}(\te,\s)).
    \end{cases}
\end{equation}

    Notice that when we consider only $0$-th order terms in $Z^{\bf v}$ and $I^{\bf v}$, the last term in equation for $\dot{\l}_{\s}$ is zero.
    
    To solve (\ref{HamPA}), we need to first apply the shooting algorithm described in Section~\ref{sec:method} in order to find the initial conditions $(\l_\te(0),\l_\s(0))$.

    For the inhibitory population model \eqref{Inhib} with the parameters in \eqref{eq:Iparam}, the system has a periodic orbit with Floquet multipliers $m_i=e^{\mu_i T}$, that are real and distinct $m_1=0.15$ and $m_2=5 \cdot 10^{-5}$. Thus, the phase-amplitude reduction considers the slowest contracting direction $\sigma$ associated to the eigenvalue $m_1$, with Floquet exponent $\mu_1=-5.44 \cdot 10^{-2}$.
    In Figure \ref{fig:alpha}(a) we show the optimal control $u$ obtained by solving the control problem \eqref{OCP2} for this model with 0-th order approximation for functions $Z^{\bf v}(\te,\s) \approx Z_0^{\bf v}(\te)$ and $I^{\bf v} (\te,\s) \approx I_0^{\bf v} (\te)$, $t_f=0.8T$ and two different values of the weight $\alpha$ in the cost function. One can observe that the controlled trajectory for sufficiently large values of $\alpha$ stays closer to the original limit cycle. This can be visualized in the evolution of the $\sigma$ variable (assessing the distance to the limit cycle), which in the case $\alpha=12$ takes values closer to 0 compared with $\alpha=0$ (see Figure \ref{fig:alpha}(b)).
    
    For the E-I system \eqref{Epop}-\eqref{Ipop} with parameters in \eqref{eq:EIparam} the system has a periodic orbit with Floquet multipliers $m_i=e^{\mu_i T}$, that are given by
    $m_1=5.4 \cdot 10^{-2}$, $m_2=\bar{m}_3=(2.3 +3.1\, i)\cdot 10^{-4}$, $m_4=-1.57 \cdot 10^{-10}$ and
    $m_5=-3.99 \cdot 10^{-10}$. Thus, the phase-amplitude reduction considers the slowest contracting direction $\sigma$ associated to the eigenvalue $m_1$, with Floquet exponent $\mu_1=-0.14$. Notice that although the other Floquet exponents are complex, the smallest one in modulus is real.
    In Figure \ref{fig:alpha}(c) we show the optimal control obtained by solving problem \eqref{OCP2}  with first order approximation for functions $Z^{\bf v}(\te,\s) \approx Z_0^{\bf v}(\te) + \sigma   Z_1^{\bf v}(\te)$ and $I^{\bf v} (\te,\s) \approx I_0^{\bf v} (\te) + \sigma   I_1^{\bf v}(\te)$, and $\alpha=0$ (without amplitude penalization) and $\alpha=0.05$. As one can observe, by including the amplitude penalization the $\sigma$ variable remains closer to zero (see Figure \ref{fig:alpha}(d)). By adding higher order terms in $\sigma$, we obtain a better precision in the description of the phase dynamics. 
    
    There is some arbitrariness in the choice of scale for the sigma variable, and this affects the choice of the parameter $\alpha$ in the cost function. Notice the difference in scales for the $\alpha$ of the inhibitory population model and the E-I network model.

\begin{figure}
\begin{tabular}{ll}
(a) & (b) \\
\includegraphics[scale=0.35]{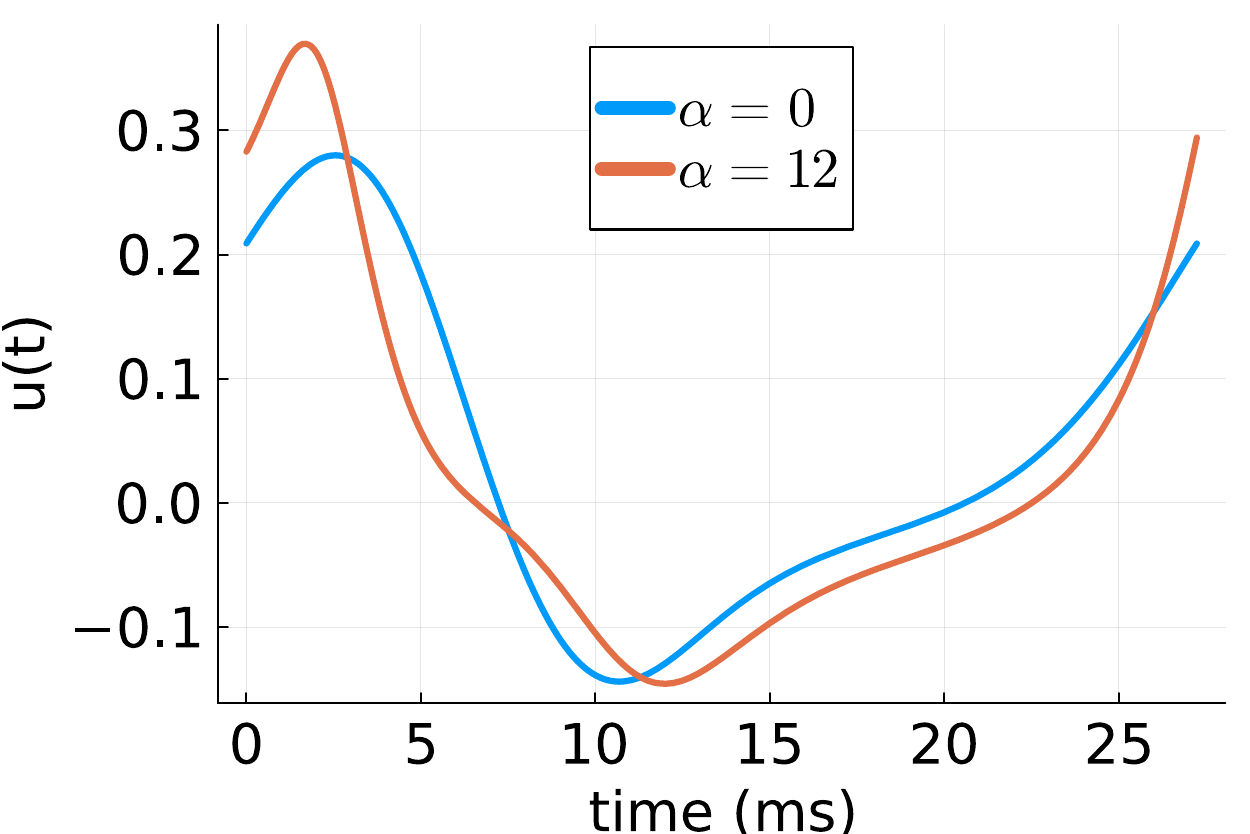} &
\includegraphics[scale=0.32]{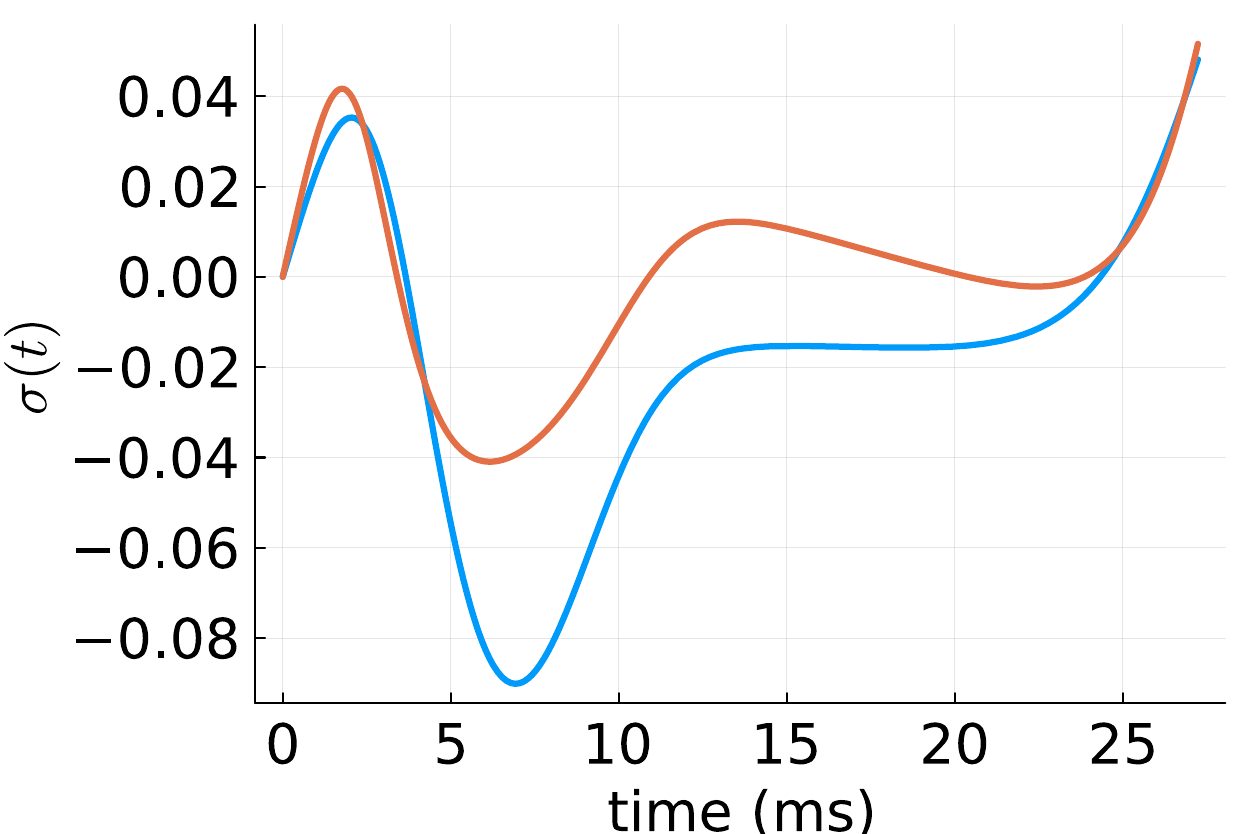}\\
(c) & (d) \\
\includegraphics[scale=0.35]{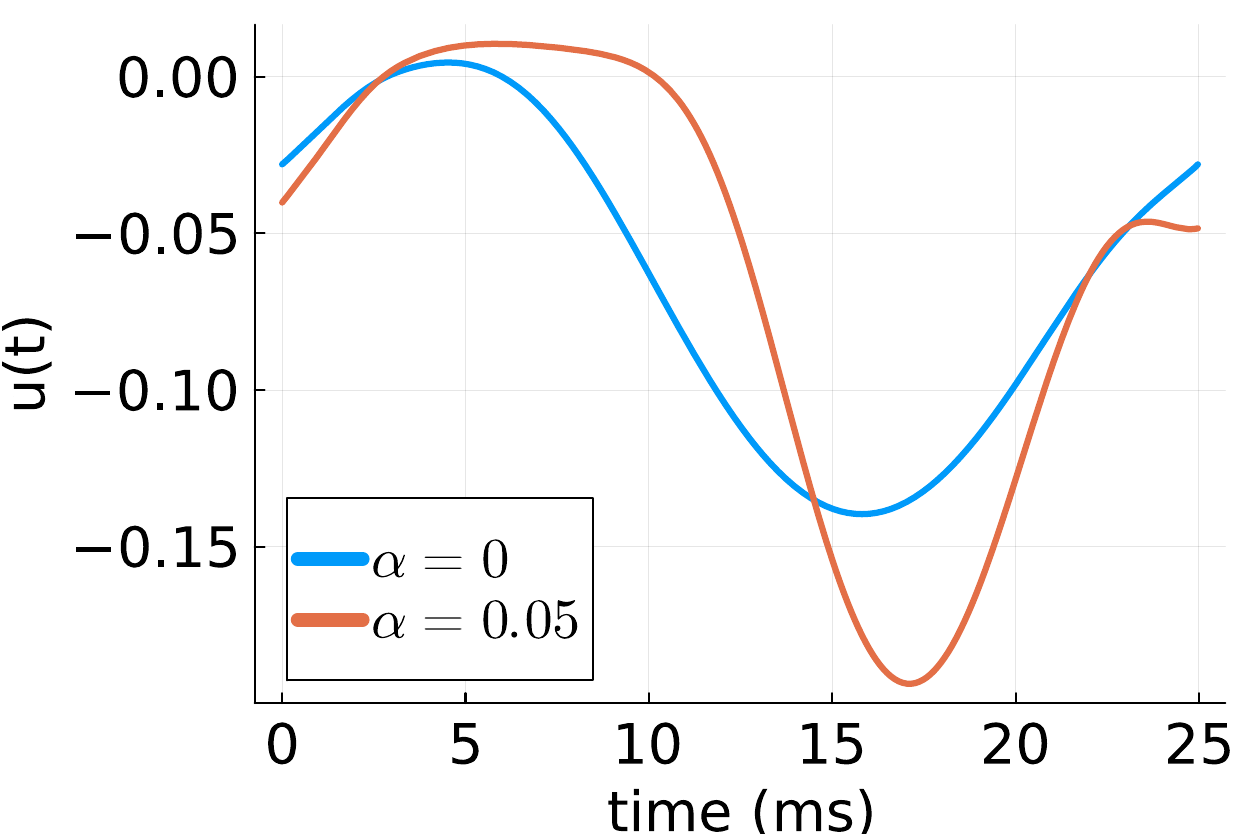} &
\includegraphics[scale=0.32]{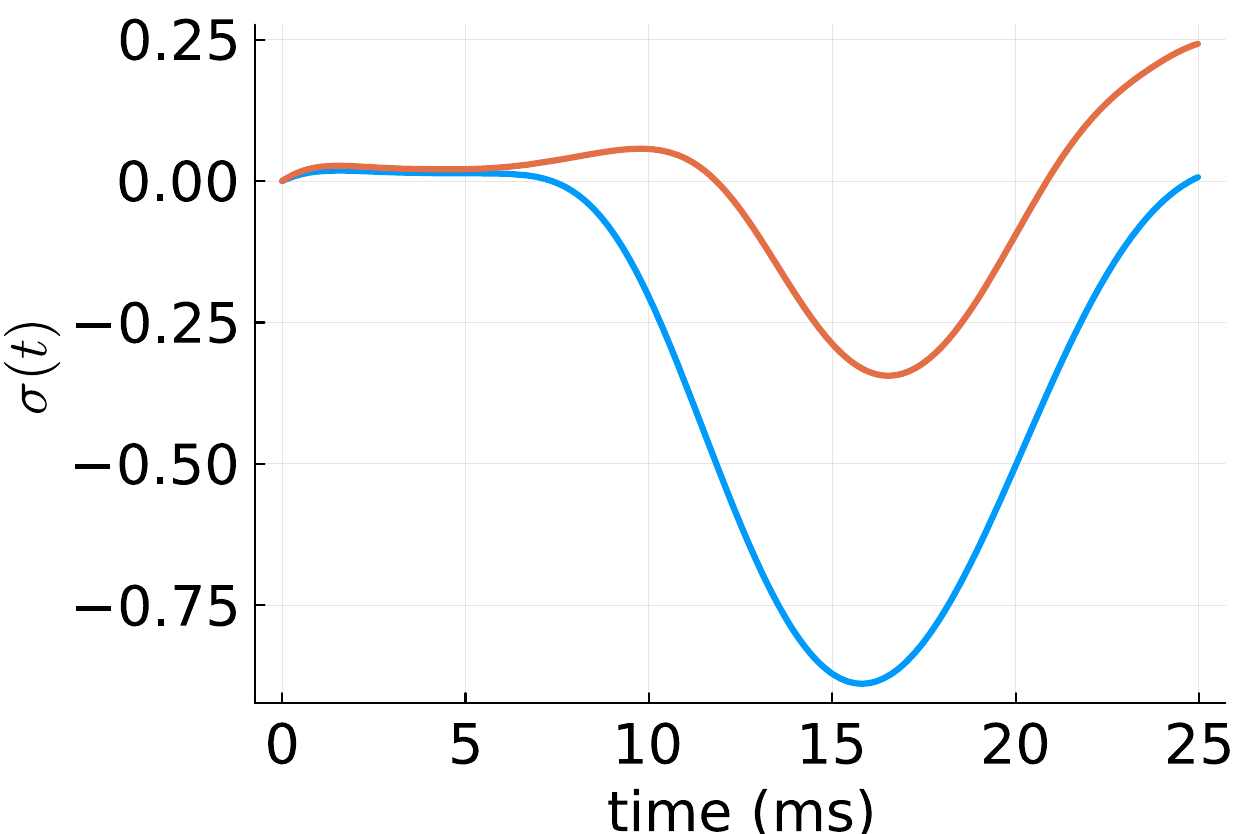}\\
\end{tabular}
\caption{(a,b) For the inhibitory population model \eqref{Inhib} and the set of parameters given in \eqref{eq:Iparam} (a) optimal control $u$ obtained by means of solving \eqref{OCP2} with 0-th order approximation with $t_f=0.8T$ and two different values of the parameter $\alpha$ in the cost function: $\alpha = 0$ (blue curve) and $\alpha=12$ (red curve); (b) Time series of the $\sigma$ coordinate corresponding to the slowest contraction direction along the controlled trajectory $x_u(t)$ (c,d) For the E-I population model\eqref{Epop}-\eqref{Ipop} and the set of parameters given in \eqref{eq:EIparam}, (c) optimal control $u$ obtained by means of solving \eqref{OCP2} with first order approximation in $\sigma$ for functions $Z^{\bf v}$ and $I^{\bf v}$, $t_f=1.2T$ and two different values of the parameter $\alpha$ in the cost function: $\alpha = 0$ (blue curve) and $\alpha=0.051$ (red curve); (d) time series of the $\sigma$ coordinate corresponding to the slowest contraction direction along the controlled trajectory $x_u(t)$ 
}
\label{fig:alpha}
\end{figure}

\subsection{Numerical implementation of the method}
\label{sec:method}

 In this section we discuss the numerical implementation of the method described in the previous sections for systems \eqref{Inhib} and \eqref{Epop}-\eqref{Ipop} that we use along this paper, but the method naturally applies to any system with a limit cycle.

 \begin{itemize}
 \item We compute the limit cycle, the period $T$ and the largest Floquet exponent $\mu$ by solving systems \eqref{Inhib} and \eqref{Epop}-\eqref{Ipop} with $u \equiv 0$ and parameters as given by \eqref{eq:Iparam} and \eqref{eq:EIparam}, respectively, with a Runge-Kutta method of order 4-5 and using a Newton's method on a suitable Poincaré map. The Floquet exponent is obtained from the monodromy matrix. 

 \item We compute the iPRC $Z_0$ and iARC $I_0$ by computing the resolventof the differential equations and taking the eigenvector of the eigenvalue 1. This is the initial condition that provides the periodic orbit that is obtained by solving the differential equations in \eqref{eq:eqPRC} and \eqref{eq:eqARC}, respectively, with the same Runge-Kutta method. The functions $Z_0$ and $I_0$ are obtained on an adaptive grid, and we interpolate them using Hermite's polynomials. 
 
 \item We compute the first order terms $Z_1$ and $I_1$ by using approximations of the parameterization $K$ in \eqref{Kparam} computed using the algorithms in \cite{Perez_Cervera20} (the formula for the terms $Z_1$ and $I_1$ in terms of the coefficients of $K$ is described in detail in Appendix C in \cite{Perez_Cervera20}). The functions are obtained in a discretized grid and we interpolate them with cubic splines. 
 
 \item We compute numerically the exponential map of Definition \ref{expo} by integrating the Hamiltonian system \eqref{HamP} (resp. \eqref{HamPA}) for problem \eqref{OCphase} (resp. \eqref{OCP2}) using a Runge-Kutta method of order 4-5 \cite{DiffEqjl}.
 
  \item Use Newton's method to find the zeros of the shooting function. The derivatives of the shooting function can be obtained either using automatic differentiation or finite differences, and both methods work well. At each iteration of the method, the exponential map is computed as in the previous step. The zeros of the shooting function provide the initial condition for $\lambda_{\theta}$ to solve the Hamiltonian system \eqref{HamP} (resp. $\lambda_{\theta}$ and $\l_{\s}$ to solve the Hamiltonian system \eqref{HamPA}). The solution is stored in an adaptive grid provided by the ODE numerical solver. This process gives us the optimal control in the same adaptive grid (or, to be more precise, the extremal control) using equations \eqref{eq:uoptimal1} (resp. \eqref{eq:uoptimal2}).

  \item Finally, we plug the computed control $u$ into the system (\ref{Inhib}) for the self-inhibitory population and (\ref{Epop})-(\ref{Ipop}) for the E-I population in order to obtain the controlled extremal solution in the original variables. Since the function $u$ is discretized,  we interpolate using Hermite's polynomials to integrate the system using a Runge-Kutta method of order 4-5.
\end{itemize}

\section{Applications to Communication Through Coherence}
\label{sec:CTC}

In this section, we test our methodology to study how the control can help to establish communication between two oscillating neuronal groups in the context of communication through coherence (CTC) theory. As explained in the introduction, the CTC theory suggests that two oscillating neuronal groups communicate effectively when they are properly \emph{phase locked} so that the presynaptic periodic input volleys arrive at the peaks of excitability of the postsynaptic group (receiving population) or, equivalently, at the phase of  minimum inhibition. In this context, we use the control as a top-down mechanism capable of delivering a specific input to align the receiving population with the optimal phase to establish communication with a given presynaptic population. To model the input from the presynaptic neural group, we introduce a periodic input in the form of successive bursts of excitatory current to the target network, that we will refer as the \emph{primary input}. The target network is modelled by means of the E-I network model introduced in the previous section. In this context, we say that communication between two oscillating neuronal groups is established if there is an amplification of the firing rate of the postsynaptic population due to the external input, while the magnitude of the amplified response is modulated by the input strength.  

We work with two different settings. First, we perturb the E-I population with one input in an adverse scenario for communication, namely, when the period of the input is equal or larger than the natural period of the E-I population. We design a control so that the input can establish communication with the target network. In the second scenario, we add a distractor as a new input and we probe whether the control is capable to maintain the communication with the primary while ignoring the distractor. 

Each input signal $p_j(t)$ to the target population will be a periodic function modeled by a von Mises probability density function in order to mimic realistic inputs in the cortex, where the input volleys are concentrated around some phases of the cycle. Mathematically, we define the $T_j$-periodic input $p_j(t)$ as
\begin{equation}\label{eq:vm_input}
p_j(t)= A_j \frac{\exp\left( \kappa \cos (\frac{2 \pi}{T_j} (t - \nu))\right)}{ I_{0}(\kappa)}, \text{ for } t \in [0,T_j),
\end{equation}
where $I_0=\frac{1}{T_j}\int_{0}^{T_j}\exp(\kappa \cos (2\pi\,(t-\nu)/T_j)dt$. That is, the temporal average over one period is $A_j$.
The parameters $\kappa$ and $\nu$ control the width and position of the peak of the input volley, respectively. In this paper, we have chosen $\kappa = 12$, to get a highly coherent signal (small width), closer to what can be found experimentally. 
This input $p_j(t)$ enters into system \eqref{Epop}-\eqref{Ipop} through the terms $I_e(t)$ and $I_i(t)$. More specifically, we take 
\begin{equation}\label{eq:Ik}I_{k}(t)=\bar{I}_{k}+ \tau_k\sum\limits_{j} p_j(t), \quad k \in \{e,i\}
\end{equation} 
where $\bar{I}_{k}$ are tonic currents. Notice the time constant $\tau_{e,i}$ multiplying the periodic inputs $p_j(t)$. 

Let us define some notation first. The solution of the original system \eqref{Epop}-\eqref{Ipop} without external perturbation $p$ or control $u$ will be denoted by ${\bf x}$, the solution of the controlled system (when no other perturbation is applied except the control $u$) will be denoted by ${\bf x}^u$, the solution of the perturbed system (in the absence of control) will be denoted by ${\bf x}^p$ and the solution of the controlled system with the perturbation will be denoted by ${\bf x}^{up}$. The subscripts will be applied to all the variables of the vector $\bf{x}$, i.e. ${\bf x}^k=(r^k_e, V^k_e, S^k_{ei}, S^k_{ee}, r^k_i, V^k_i, S^k_{ie}, S^k_{ii})$, for $k\in\{\emptyset,u,p,up\}$. 

In order to establish communication between two populations, the input from the presynaptic population phase-locks with the adequate phase with the oscillatory activity of the target population, so that the presynaptic input produces an effect on the firing rate of the postsynaptic population. Moreover, in \cite{ReynerHuguet22} (as also suggested in \cite{SarafY19}) it is found that changes in the input strength $A$ are transmitted and better reproduced at the output by the spike synchronization properties of the E-population (reflected through both the maximum firing rate and half-width of the E-volley) rather than by the average firing rate $r_e$. For this reason we measure changes in the firing rate of the postsynaptic group in the spike synchronization.

Following \cite{PCSH_cnsns20,ReynerHuguet22} we chose two factors to quantify communication: the \emph{synchronization index} (SI) $\rho$, which measures the synchronization or coherence of the presynaptic and postsynaptic groups and the \emph{amplification factor} $\Delta\alpha$, which measures the amplification of the firing rate of the postsynaptic group due to the external input. Next, we provide a precise definition of these two concepts.

To measure synchronization of the E-I network with a $T_j$-periodic input $p_j(t)$, we consider the values of the phase variable $\theta$ at integer multiples of the period $T_j$, i.e $\theta_k=\theta(t_0+k T_j) \in \mathbb{T}$, with $k \in \mathbb{N}$. Recall that the time evolution of the phase variable $\theta$ is provided by system \eqref{eq:PA} (where we have considered $\nabla \Theta$ and $\nabla \Sigma$ approximated up to first order in $\sigma$). Thus, the synchronization index $\rho$,  also known as {\em vector strength} or {\em Kuramoto order parameter}, \cite{Pikovsky01}, is a measure of how clustered are the events over a cycle and is computed according to the following formula,
\begin{equation} \label{eq:SI}
z=\rho e^{i \phi}=\frac{1}{N}\sum_{j=1}^N e^{i \theta_j}, \qquad \rho=|z|.
\end{equation}
Notice that perfect clustering is obtained when $\rho=1$, whereas if phases are scattered around the circle, then $\rho \approx 0$.

We also define the factor $\Delta\alpha$ as the rate change of the maximum of the firing rate E-volley due to the external perturbation. Mathematically, 
\begin{equation} \label{eq:Dalpha}
    \Delta\alpha = \frac{\frac{1}{N} \sum_{i=1}^{N} r_e^{up}(t_i)}{ \bar{r}^u_e}, 
\end{equation}
where $\{t_i\}_{i=1}^N$ correspond to the times where local maxima of $r_e^{up}(t)$ are attained for $N$ cycles, and $\bar{r}^u_e$ is the maximum of the excitatory firing rate when only the control is applied. We recall that in the latter case we obtain a periodic orbit.

\subsection{Control-induced communication for a single input}\label{sec:SingleInput}

\begin{figure}
\begin{tabular}{ll}
(a) & (b) \\
\includegraphics[scale=0.33]{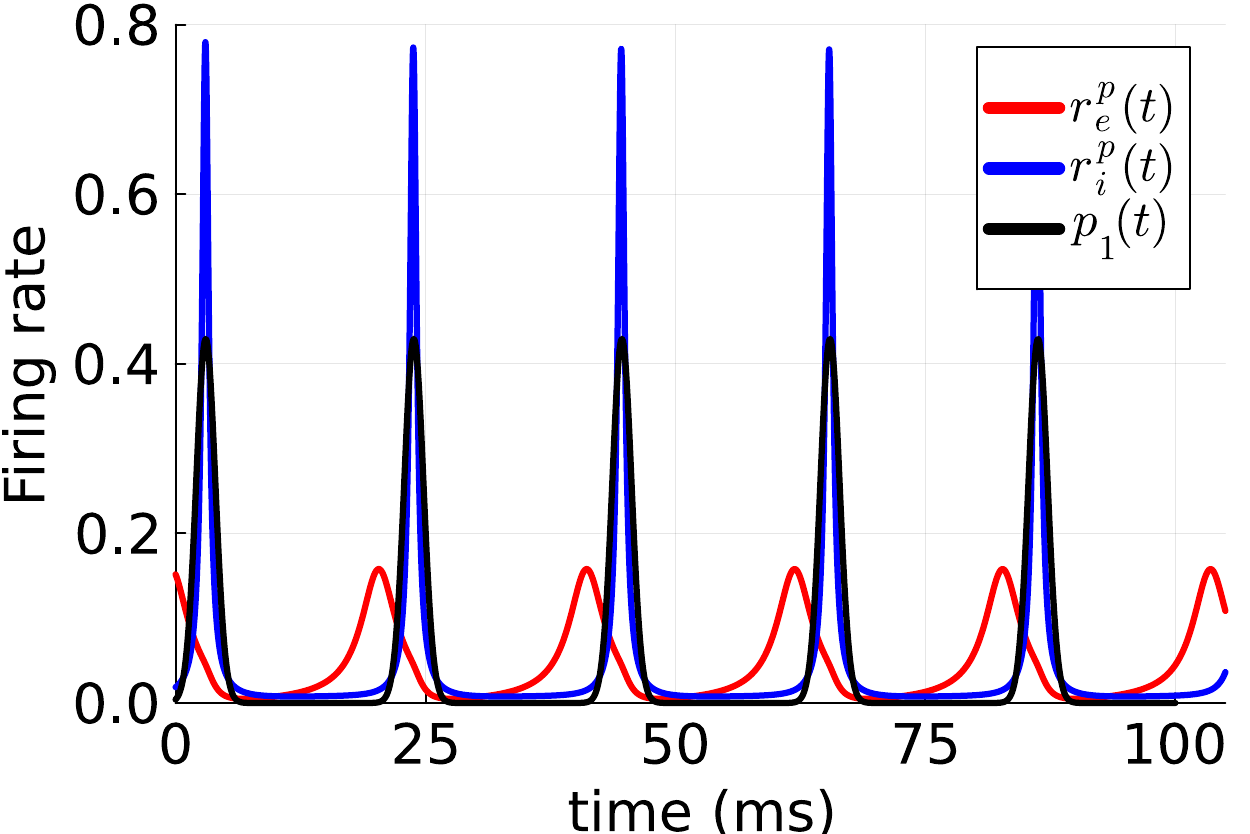} &
\includegraphics[scale=0.33]{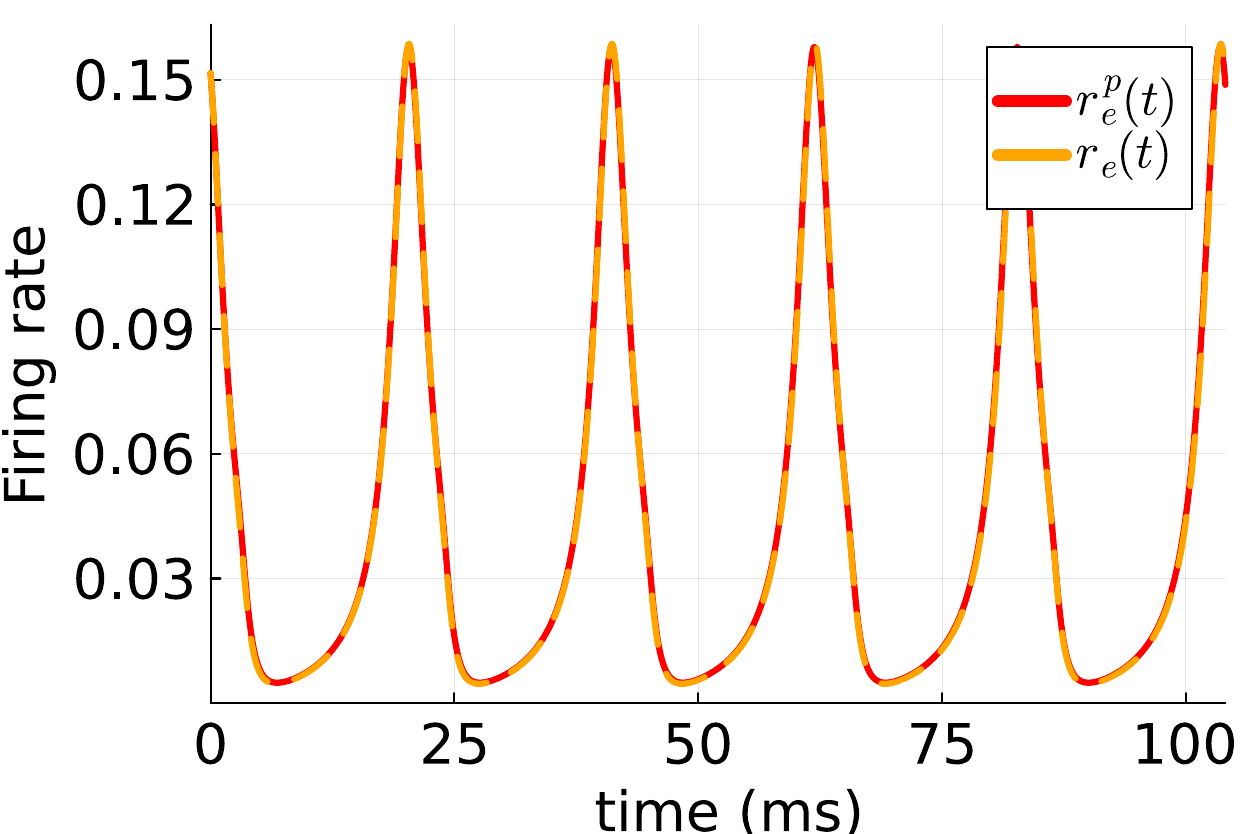} \\
(c) & (d) \\
\includegraphics[scale=0.33]{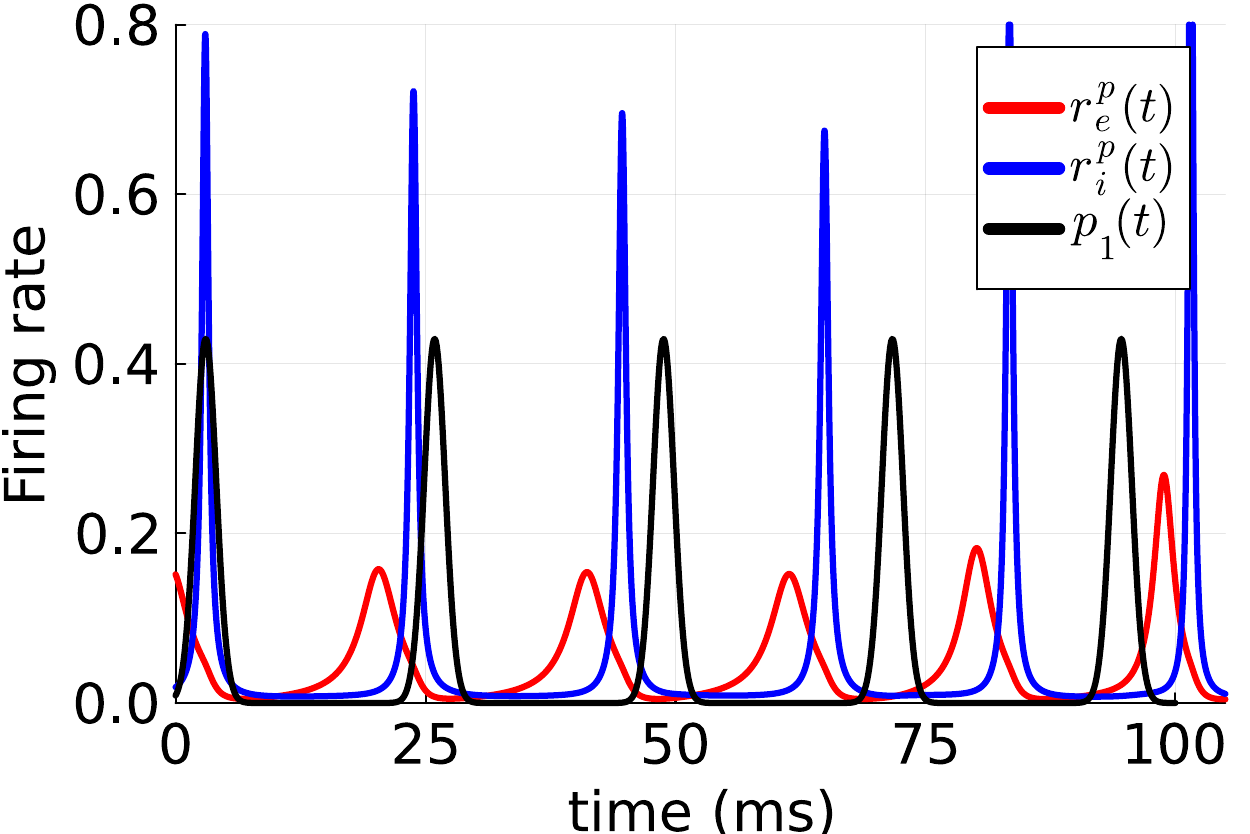} &
\includegraphics[scale=0.33]{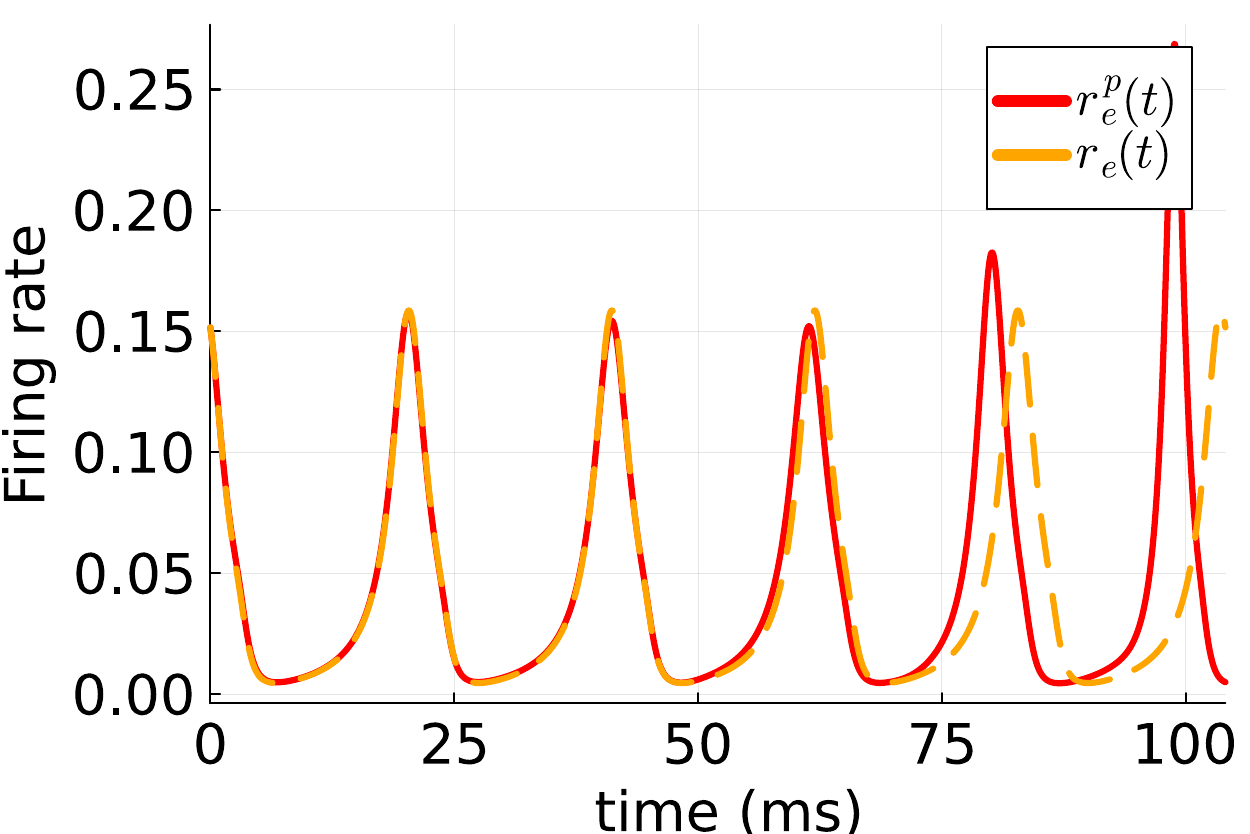}
\\
\end{tabular}
    \caption{Firing rate of the target network in the presence of a perturbation but in the absence of control. (a,c) Temporal evolution of the perturbed (but non-controlled) firing rate of the E-cells $r^p_e$ (red), I-cells $r^p_i$ (blue) and the external input $p_1$ (black curve) over a few cycles for the system \eqref{Epop}-\eqref{Ipop} with $u(t) \equiv 0$ and the perturbation $p_1$ defined in \eqref{eq:vm_input} with $A_1 = 0.05$ and (a) $T_1 = T$ and (c) $T_1=1.1 T$. 
    (b,d) Time series of the original firing rate $r_e$ without perturbation (orange) and the same $r_e^p$ as in panel (a,c), respectively. Notice that the perturbation has barely no effect on the E-firing rate when there is no control.}
    \label{fig:ctc_setting}
\end{figure}

From previous work \cite{ReynerHuguet22}, we know that inputs whose period is equal or greater than the natural period of the firing rate of the postsynaptic population are not capable to entrain the E-I network to communicate effectively; see Figure~\ref{fig:ctc_setting}. Notice that the perturbation does not phase-lock to $r_e$ when $T_1 >T$ (Figure~\ref{fig:ctc_setting}(c)) and even when there is entrainment for $T_1=T$ (Figure~\ref{fig:ctc_setting}(a)) the perturbation does not affect the firing rate of the E-cells (see that there is no difference in the firing rates $r_e$ and $r_e^p$ in Figure~\ref{fig:ctc_setting}(b)). In this section, we will show that, even in this adverse situation, an adequate control can set the target system in the optimal phase for communication with the presynpatic population.

Accordingly, we consider the external input to be a single periodic stimulus $p_1(t)$ of von Mises type \eqref{eq:vm_input}. We test three different periods: $T_1 = T$, $T_1 = 1.1T$, and $T_1 = 1.2T$, where $T$ is the period of the unperturbed cycle in the E-I network.

Given $T_1 = aT$, for $1 \leq a < 2$, our control strategy consists of choosing a value $t_f$ for the optimal-control problem \eqref{OCP2} so that the period of the target population is increased. More precisely, the period is lengthened to allow the input volleys to arrive while the inhibition has not yet been activated and is at its minimum. Mathematically,
\begin{equation}
\label{eq:tf}
t_f = a(T+t_{{r_i}_{\max}}-t_{{r_i}_{\min}}),
\end{equation}
where ${r_i}_{\max}$ and ${r_i}_{\min}$ denote the maximum and the minimum, respectively, of the firing rate of the inhibitory population $r_i$ of the unperturbed system on the limit cycle.

Then, we compute the extremal control for one period by means of solving \eqref{OCP2} with $t_f$ as in \eqref{eq:tf} for $a=1,1.1,1.2$, using the methodology described in Section \ref{sec:Ocontrol} and we apply the control periodically to the original system \eqref{Epop}-\eqref{Ipop}. Figure \ref{fig:PP}(a) shows the controls $u_1$, $u_{11}$ and $u_{12}$, computed for one cycle, repeated over several cycles. Notice that the shape is similar to the opposite of the sum of the PRCs (see Figure \ref{fig:PRCARC}(c)). Figure \ref{fig:PP}(b) shows the effect of the control on system \eqref{Epop}-\eqref{Ipop} without any other time-dependent perturbation ($I_e\equiv 10$ and $I_i\equiv 0$). The period of the controlled population becomes $T_{new} = t_f$ defined in \eqref{eq:tf}. The controlled trajectories (blue, red, green for $a=1, 1.1, 1.2$, respectively) detach from the original limit cycle (grey) and approach the limit cycle of the corresponding controlled system. Since our aim was to apply the control periodically, and this fact displaces the trajectory away from the limit cycle, we have decided not to penalize the distance to the limit cycle in the cost function, and thus, the parameter $\alpha$ is set to 0. See the discussion in section~\ref{sec:discussion} for more details.

The time-dependent control system in the original coordinates is given by system \eqref{Epop}-\eqref{Ipop} with
\[I_e(t)=\bar{I}_e + \tau_e p_1(t),\]
and 
\[I_i(t)=\bar{I}_i + \tau_i p_1(t).\]

\begin{figure}
\begin{tabular}{ll}
(a) & (b) \\
\includegraphics[scale=0.32]{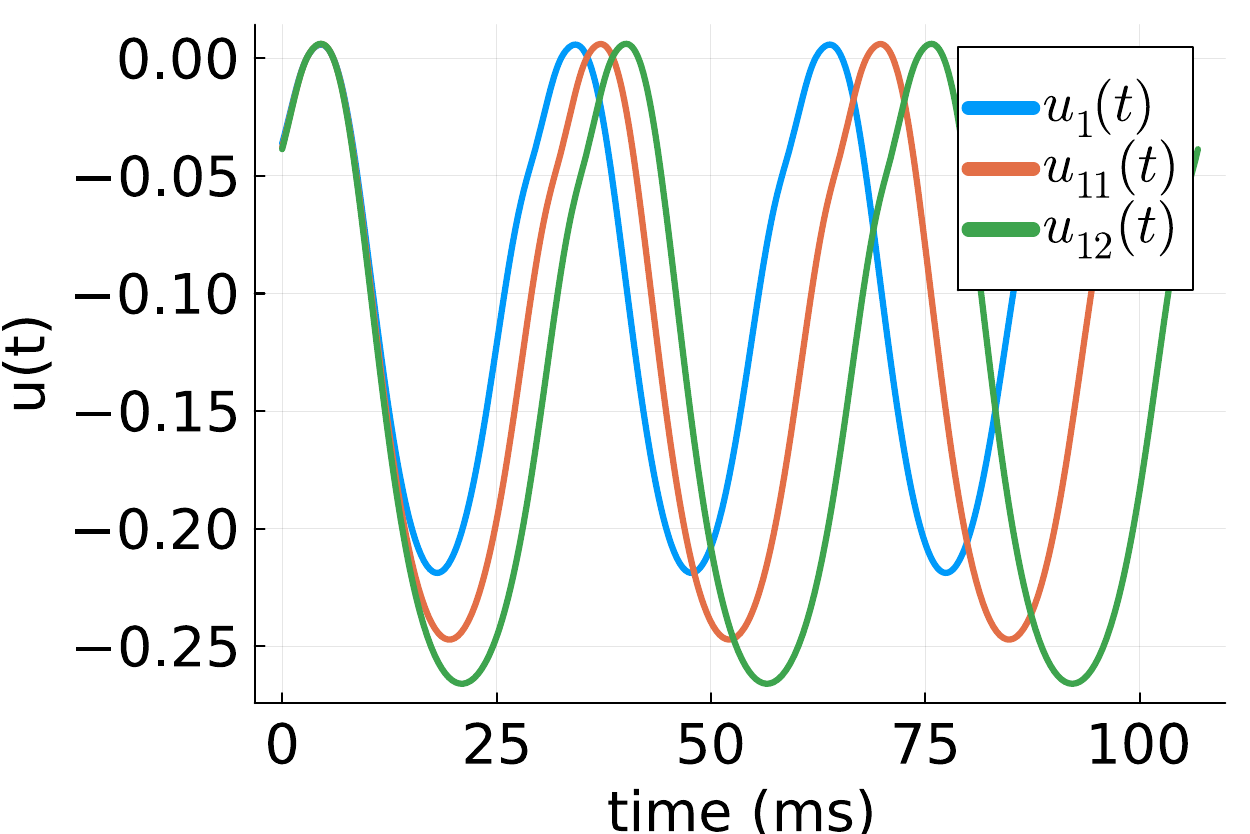} &
\includegraphics[scale=0.32]{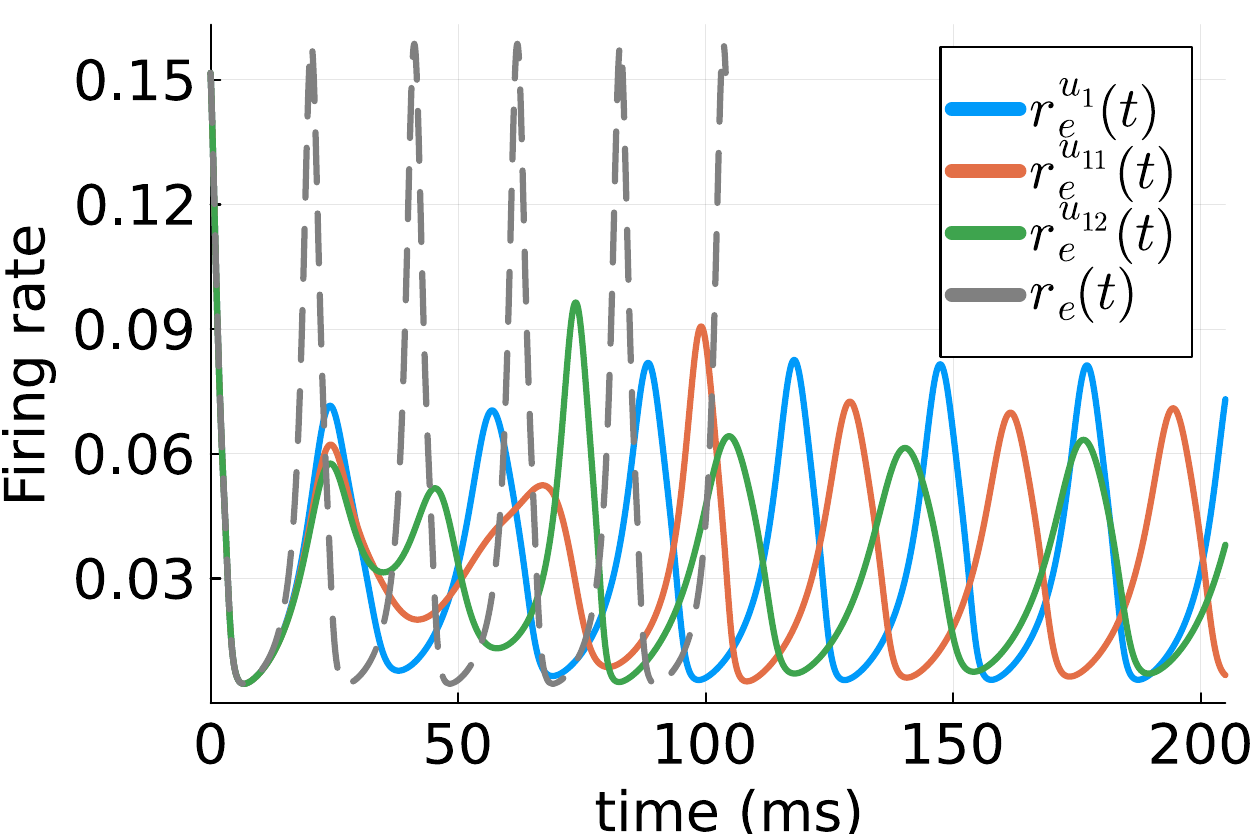} \\
\end{tabular}
\caption{(a) Time evolution of the control $u(t)$ obtained by solving the control problem \eqref{OCphase} with $t_f$ as in \eqref{eq:tf} and $a=1$ ($u_1$), $a=1.1$ ($u_{11}$) and $a=1.2$ ($u_{12}$). 
(b) Time evolution of the E-firing rate $r_e$ when the periodic controls in panel (a) are applied to the E-I system. Notice that after a transient the trajectory sets at a stable oscillatory regime of lengthened period $t_f$. The evolution of the variable $r_e$ on the limit cycle for the original E-I system of period $T=20.811$ is plotted up to time $t=100$ for comparison purposes (grey dashed curve).}
\label{fig:PP}
\end{figure}

We expect that the perturbation $p_1$ will now be able to entrain the controlled system, as opposed to the case in which the control was not present. To provide a measure of the entrainment, we compute the synchronization index \eqref{eq:SI} in each case (controlled versus non-controlled) and for three different periods $T_1$ of the perturbation that are larger than the natural period $T$ of the target network. We recall that for each period $T_1$, the control is different since the target period $T_{new}=t_f$ varies with $a$ (recall that $T_1=aT$).  

\begin{figure}
    \centering
    \includegraphics[scale=0.6]{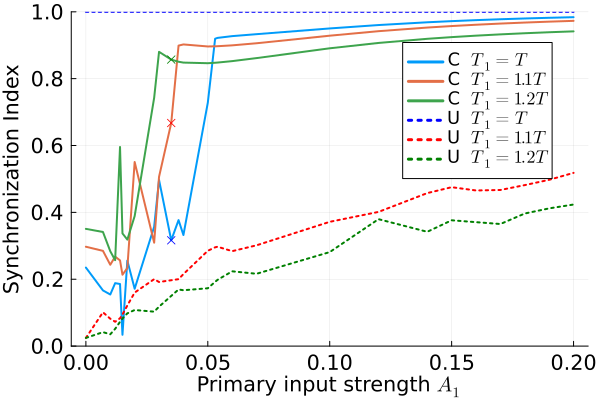}
    \caption{Synchronization index SI \eqref{eq:SI} between the E-population of the target network \eqref{Epop}-\eqref{Ipop} and the external periodic input $p_1(t)$ as a function of the amplitude of the input $A_1$ and different values of the period $T_1$ (color legend). The index is computed over $N=80$ cycles using the stroboscopic map. Dashed curves correspond to the non-controlled system ($u\equiv 0$) while solid curves correspond to the controlled system ($u$ designed according to \eqref{eq:tf}). The crosses indicate the value of the amplitude $A_1=0.035$ plotted in Figure \ref{fig:threeg1inp}.}
    \label{fig:sync1input}
\end{figure}

Figure \ref{fig:sync1input} shows the synchronization index $\rho$ defined in \eqref{eq:SI} between the E-I network (postsynaptic group) with a $T_1$-periodic input $p_1$ (presynaptic group) as a function of the input strength $A_1$ for different values of $T_1$. Notice that the presence of the control allows the perturbation to well entrain the target system for strong enough inputs (synchronization index value approaches $\rho=1$ when $A_1$ increases). Compare the curves with control (solid curves) and without it (dashed) in Figure \ref{fig:sync1input}. 
We can observe that there is a threshold for the amplitude of the input above which the postsynaptic population almost phase-locks with the input (notice the sudden jump in the curves of Figure \ref{fig:sync1input}, with values of SI approaching 1). Somehow counter-intuitively, the larger the period of the primary input, the lower is this threshold, while the synchronization index caps at a lower value for larger values of $T_1$.

One potential interpretation for this finding could be that the control, $u$, exhibits predominantly inhibitory behavior (as evidenced by Figure \ref{fig:PP}(a)). As $T_1$ increases, $u$ becomes even more inhibitory, leading to a lengthening of the cycle and the suppression of inhibitory neurons. This, in turn, amplifies the response to the external input and facilitates synchronization, particularly when the input strength $A_1$ is weak. Upon reaching a certain threshold, the strength of the input $A_1$ can overcome the inhibition from the control input $u$. 
Consequently, inputs with shorter periods $T_1$ become more effective at entraining the network because they can outpace the natural activation of the I-cells in the network \cite{ReynerHuguet22}.
Despite the control input is slowing down the activation of the I-cells to a rate lower than that of the external primary input, the external input can still interfere with the effects of the control input, particularly when the external input is much slower. This interference can have detrimental effects on entrainment, and explain why SI is lower for larger values of $T_1$ (when $A_1$ is large).

To illustrate this explanation, in Figure \ref{fig:threeg1inp}(top) we show the time series of the firing rate of the E and I populations for three representative cases corresponding to the crosses in Figure \ref{fig:sync1input}. Notice that when the input volley arrives prior to the activation of the I-cells, it can trigger a response in the target system. Thus, the time difference between the peak of the perturbation and the peak of the inhibitory firing rate is determinant for communication; in particular, if these peaks match, the input  will simply be inhibited and communication will not be established. In order to quantify this phenomenon, we consider the sequences $\{t_{r_i}^{up}(k)\}_{k}$, where $t_{r_i}^{up}(k)$ is the time of the $k$-th local maximum of $r_i^{up}$, and $\{t_{p_1}(k)\}_{k}$, where $t_{p_1}(k)$ is the time of the $k$-th local maximum of $p_1$. From these two sequences, we define 
\begin{equation}\label{eq:deltatauk}
\Delta\tau_k:=t_{r_i}^{up}(k)-t_{p_1}(j_k),
\end{equation}
where $j_k=\underset{j\in\{1,\dots,N_C\}}{\textrm{argmin}}|t_{r_i}^{up}(k)-t_{p_1}(j)|$ and $N_C$ is the number of cycles of the perturbation. 

In Figure \ref{fig:threeg1inp}(bottom), we show the histograms of $\Delta\tau=\{\Delta\tau_k\}_k$ for the three cases shown in Figure \ref{fig:threeg1inp}(top). For the computations, we used a  simulation of 1500\,ms (in Figure \ref{fig:threeg1inp}(top) only a representative time window is shown). We clearly observe a more uniformly distributed $\Delta\tau$ histogram along a cycle for Figure~\ref{fig:threeg1inp}(d) (corresponding to $T_1=T$), which translates to a lower synchronization index. In contrast, in Figure \ref{fig:threeg1inp}(f) (corresponding to $T_1=1.2T$) we observe a more concentrated histogram between -10\,ms and -5\,ms, which we can deduce that it corresponds to the phases of the cycle with higher excitability. Figure~\ref{fig:threeg1inp}(e), corresponding to the medium synchronization index, shows a transition between the previous two situations. 

\begin{figure}
\begin{tabular}{lll}
(a) & (b) & (c) \\
\includegraphics[width=0.33\textwidth]{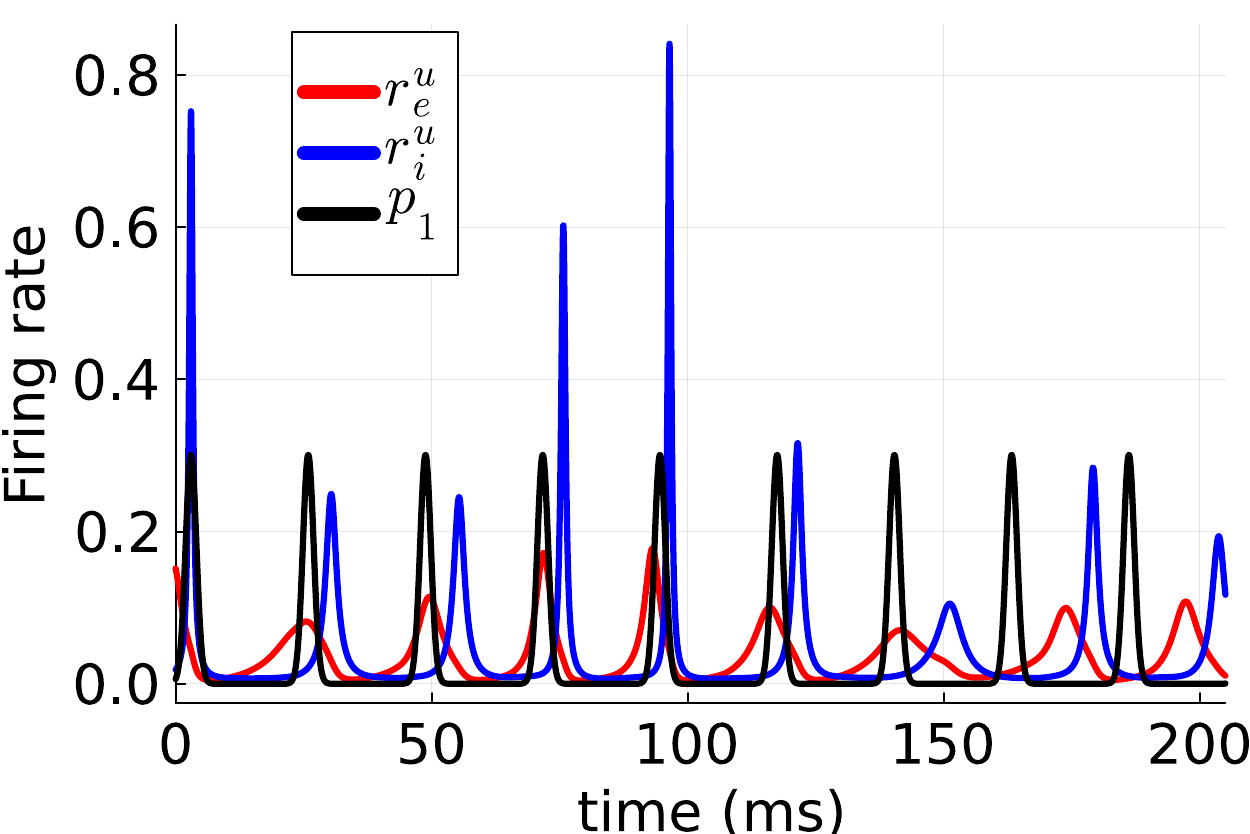} &
\includegraphics[width=0.33\textwidth]{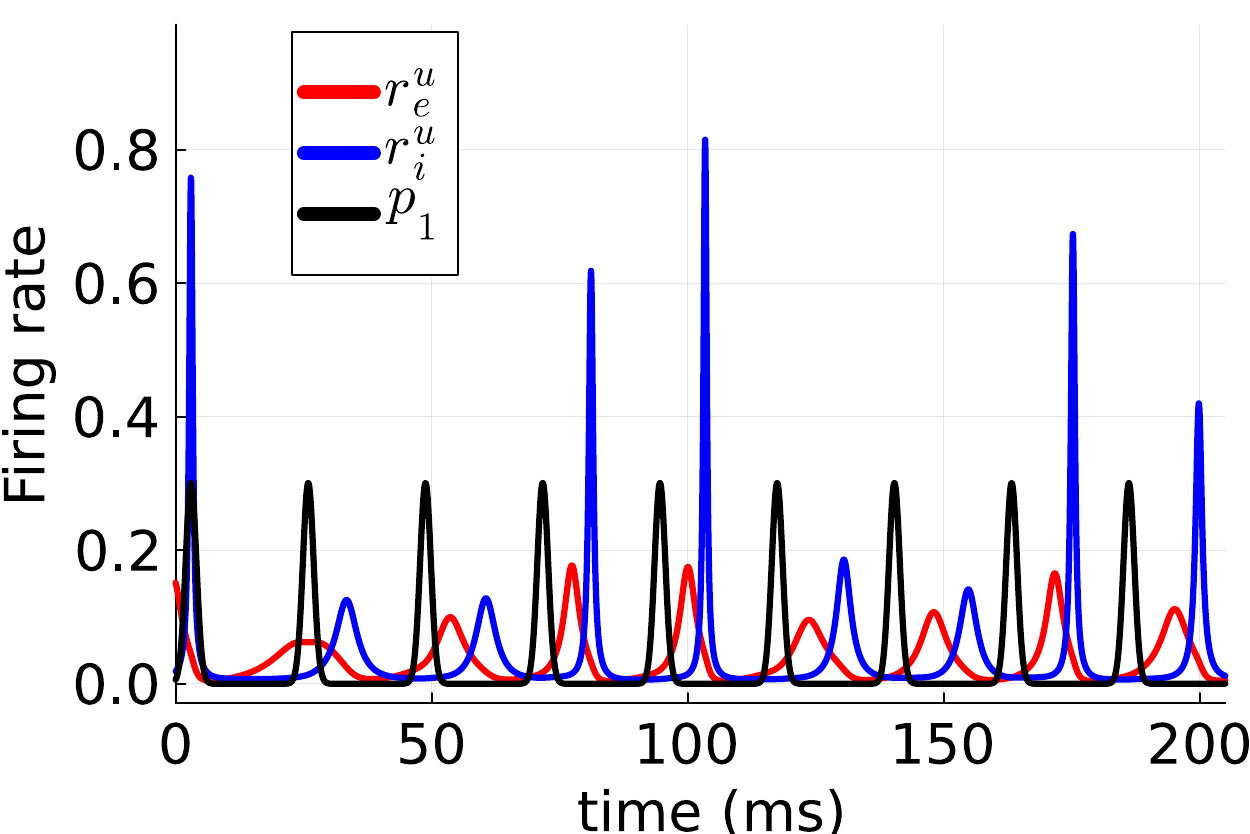} &
\includegraphics[width=0.33\textwidth]{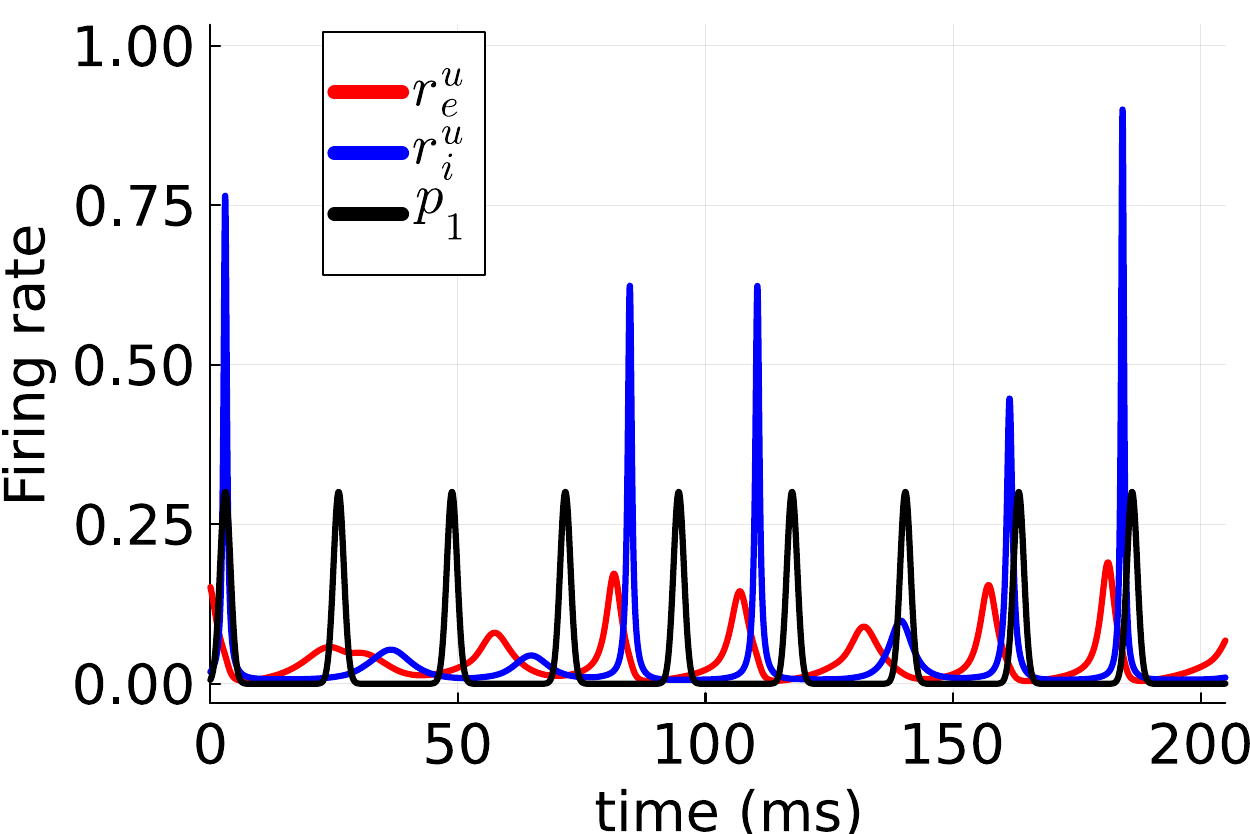} \\
(d) & (e) & (f)\\
\includegraphics[width=0.33\textwidth]{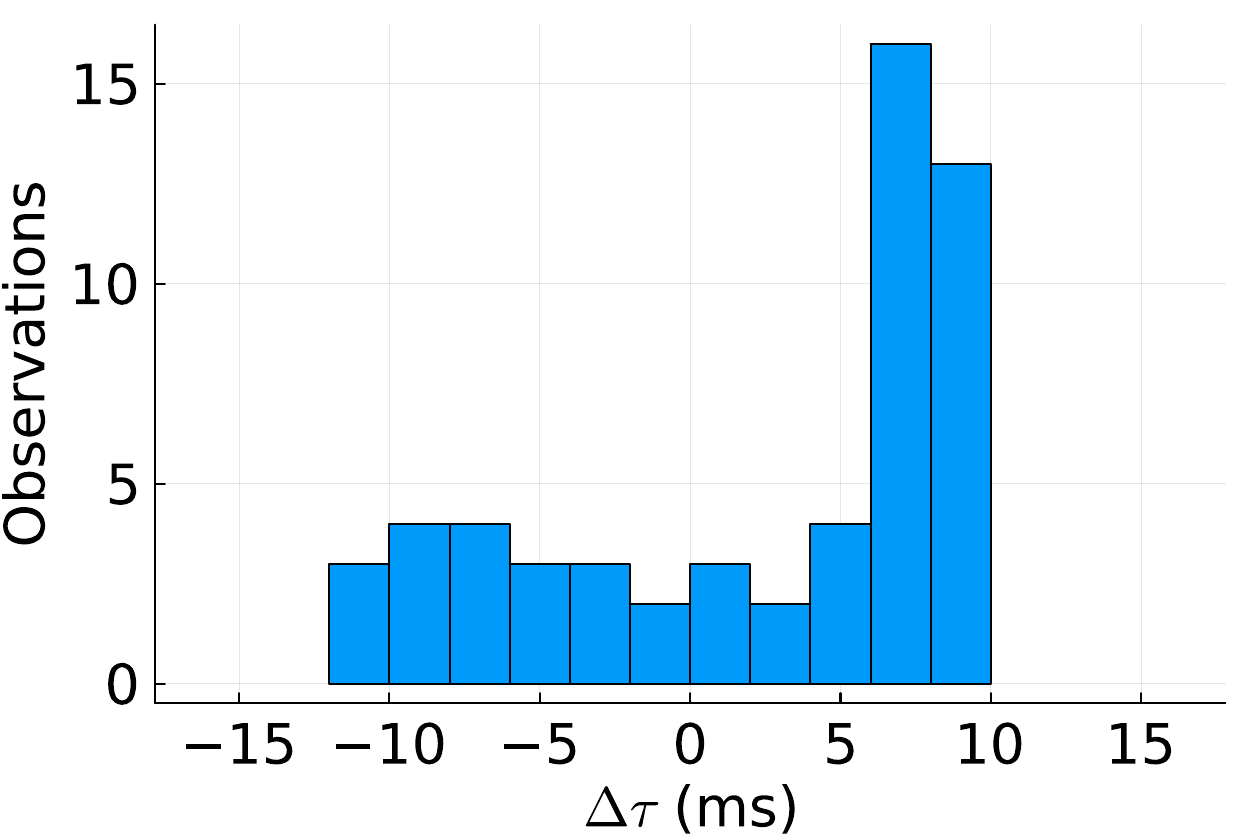} 
&
\includegraphics[width=0.33\textwidth]{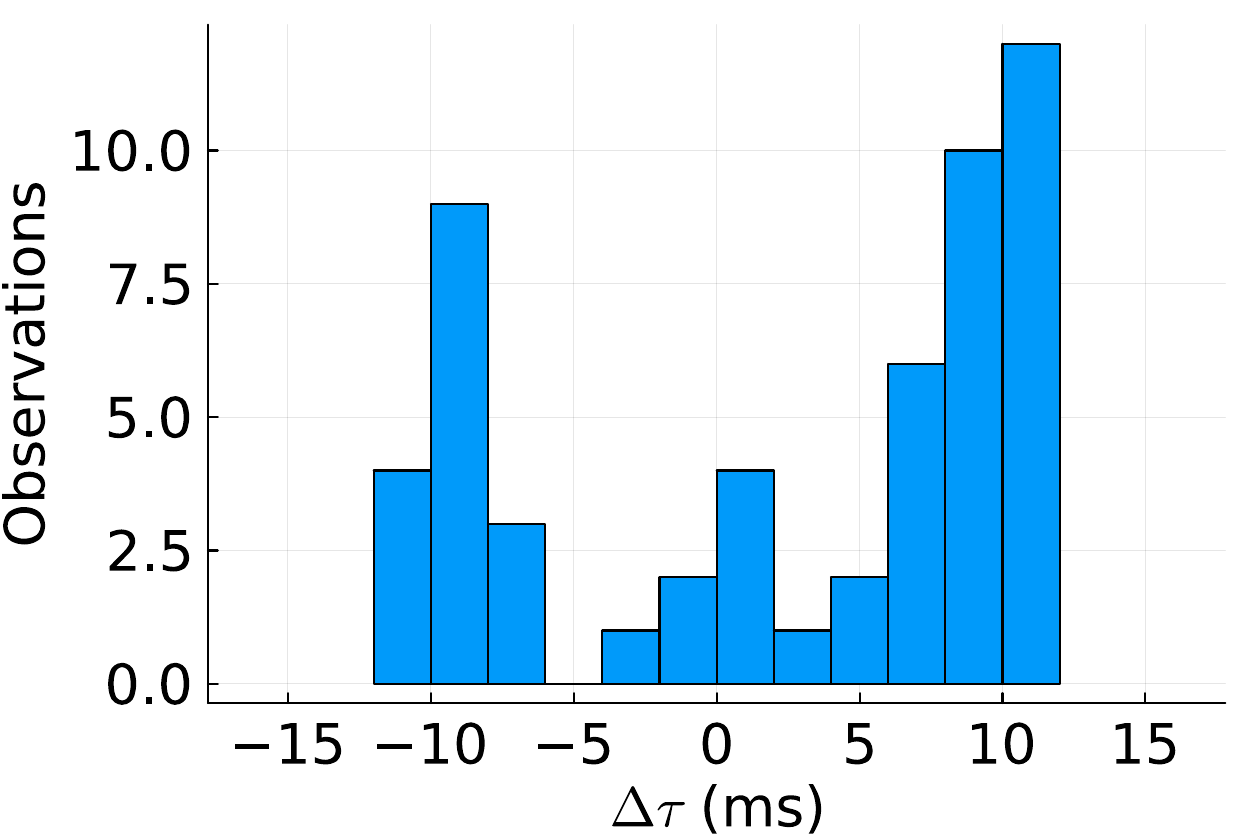} 
&
\includegraphics[width=0.33\textwidth]{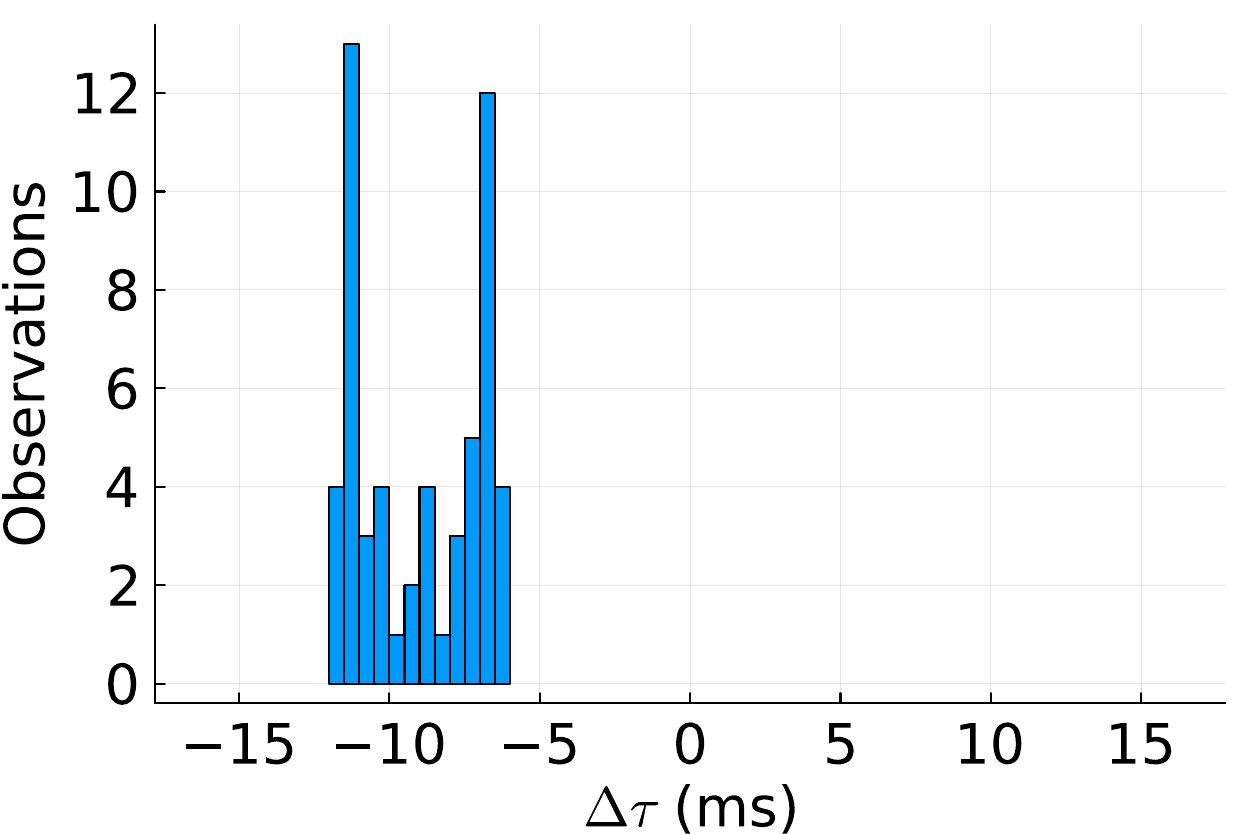}
\\
\end{tabular}
\caption{ (a,b,c) Time series of the firing rates of the E and I cells of the perturbed controlled system \eqref{Epop}-\eqref{Ipop}, $r_e^{up}$ (red), $r_i^{up}$ (blue) and the external input $p_1(t)$ (black) and (d,e,f) distribution of the time intervals $\Delta{\tau_k}$ defined in \eqref{eq:deltatauk}, for three representative cases (low, medium and high synchronization index, respectively) corresponding to the crosses in Figure \ref{fig:sync1input}:  (a,d) $T_1 = T$, $A=0.035$, (b,e) $T_1 = 1.1T$, $A=0.035$ and (c,f) $T_1=1.2T$, $A=0.035$. }
\label{fig:threeg1inp}
\end{figure}

Phase-locking is not enough to conclude that there is communication. For instance, in case the input has the same period than the target population, there is phase-locking  (see dashed blue curve in Figure~\ref{fig:sync1input}), but we know from \cite{ReynerHuguet22} that the input does not communicate with the target population, that is, there are no changes in the firing rate of the target population due to changes in the amplitude of the input. 

To be able to effectively measure the communication between the input and the target network, we compute the factor $\Delta \alpha$, defined in \eqref{eq:Dalpha}, which measures the amplification of the E-firing rate due to the external input. To compute $\Delta \alpha$ in the absence of control, we modify the formula in \eqref{eq:Dalpha} replacing $r_e^{up}$ by $r_e^p$ and $\bar r_e^u$ by $\bar r_e$.

\begin{figure}
    \centering
    \includegraphics[scale=0.5]{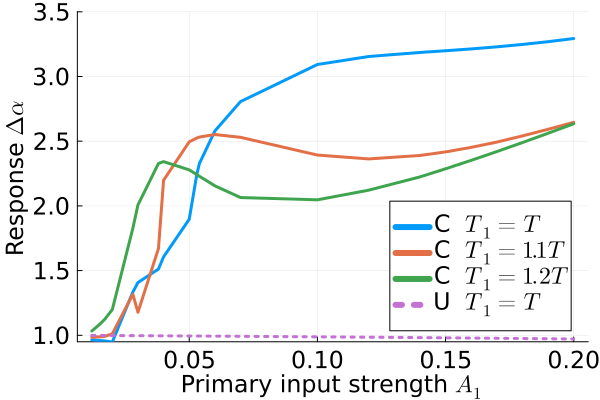}
    \caption{The factor $\Delta\alpha$ defined in \eqref{eq:Dalpha} for the perturbed controlled system \eqref{Epop}-\eqref{Ipop} with an external periodic input $p_1$ as a function of the amplitude of the input $A_1$ for three different values of the input period, $T_1=T$ (blue), $T_1=1.1 \, T$ (orange) and $T_1= 1.2 T$ (green). For the uncontrolled case, $\Delta \alpha$ is only computed for  $T_1=T$, since there is no synchronization for the other periods. The computation of $\Delta \alpha$ is only performed if the SI is above 0.8, which explains why the curves start at different values of $A_1$}
    \label{fig:responses1}
\end{figure}

In Figure \ref{fig:responses1} we illustrate the amplification of the response of the E-cells for the three different input periods $T_1$ in Figure~\ref{fig:sync1input}. Notice that the presence of the control strongly enhances the response of the E-cells and therefore the communication: for the case $T_1=T$, $\Delta \alpha$ remains at 1 for the uncontrolled system even when the input strength is increased indicating that the target network ignores the input (see the dashed purple curve in Figure \ref{fig:responses1}) while $\Delta \alpha$ reaches values around 3 for the controlled system (see blue curve in Figure \ref{fig:responses1}). The response is particularly enhanced in the case $T_1=T$.
As in the case of the synchronization index, we observe better performance the longer the perturbation period (e.g., green curve) for low values of the amplitude $A_1$ which can also be explained from the observations drawn from Figure~\ref{fig:threeg1inp}. 
 
\subsection{Control-induced selective communication}\label{sec:DoubleInput}

In this section, we perturb the E-I network with two inputs $p_1$ and $p_2$, which we refer to as the \emph{primary input} and the \emph{distractor}, respectively. As in Section \ref{sec:SingleInput}, we design the control according to the period of the primary (see equation \eqref{eq:tf}) and we explore whether the postsynaptic population (E-I network) responds to the primary input while ignoring the distractor, thus establishing selective communication \cite{Fries15}. We know from \cite{ReynerHuguet22} that for an input with a period higher than the natural gamma cycle of the postsynaptic network the communication is not effective and it is easily disrupted by a distractor. Here we show that the control can change the situation.

We use the same control obtained in the previous section and, following \eqref{eq:Ik}, we apply the external inputs $I_e$ and $I_i$ to system \eqref{Epop}-\eqref{Ipop} given by 
\begin{equation}\label{eq:input2}
\begin{array}{rl}
I_e(t) &=\bar{I}_e + \tau_e p_1(t) + \tau_e p_2 (t), \\
I_i(t) & =\bar{I}_i + \tau_i p_1(t) + \tau_i p_2 (t), \\
\end{array}
\end{equation}
where $p_1$ and $p_2$ are modelled by periodic von Mises functions \eqref{eq:vm_input} with $\kappa=12$, $T_2=0.8T$ and varying the period of the primary $T_1$ and the amplitudes of both inputs $A_1$ and $A_2$.

We first compute the synchronization index between the postsynaptic population and a primary input $p_1$ in the presence of a distractor, for the three different periods $T_1$ already considered in the previous section (see Figure \ref{fig:syncindexes}, where each panel corresponds to a different $T_1$). Each panel shows the changes in the SI as the strength $A_1$ is varied; different colors correspond to different values of the distractor's strength $A_2$. We observe that for large values of $A_1$ the synchronization index is lower when $A_2$ increases. However, when the strength of the primary is weak (low values of $A_1$), the situation is reversed. This can be explained by the fact that, when the primary is weak, it cannot entrain the network by itself and when the distractor volleys coincide, on some cycles, with the primary input volleys, the distractor helps to elicit a response of the postsynaptic group, thus enhancing synchronization, while the distractor volleys are not affecting much the postsynaptic group when they do not coincide with the primary ones. However, once the strength of the primary input $A_1$ is large enough, the primary input is capable by itself to entrain the network and the distractor only slightly distorts the entrainment by the primary. This distorsion is, of course, more noticeable if the strength of the distractor is larger, showing a lower SI for larger values of $A_2$.

We also observe that for periods $T_1$ of the primary input equal or bigger than $1.1T$, the distractor prevents the target population from getting entrained by the primary for amplitudes $A_2\geq 0.07$ (SI is below 0.8). Compare panels (a) and (c) of Figure~\ref{fig:syncindexes}.

\begin{figure}
\begin{tabular}{lll}
(a) & (b) & (c) \\
\includegraphics[width=0.32\textwidth]{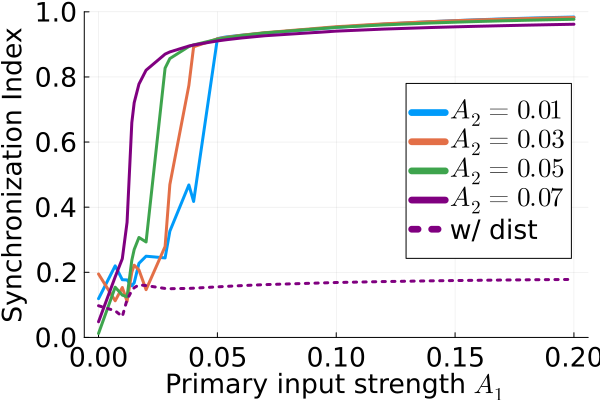} &
\includegraphics[width=0.32\textwidth]{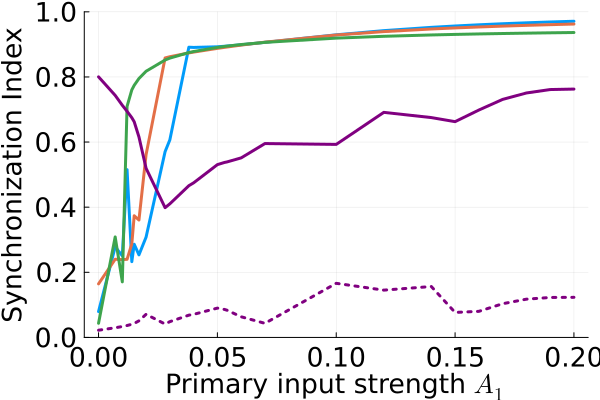} &
\includegraphics[width=0.32\textwidth]{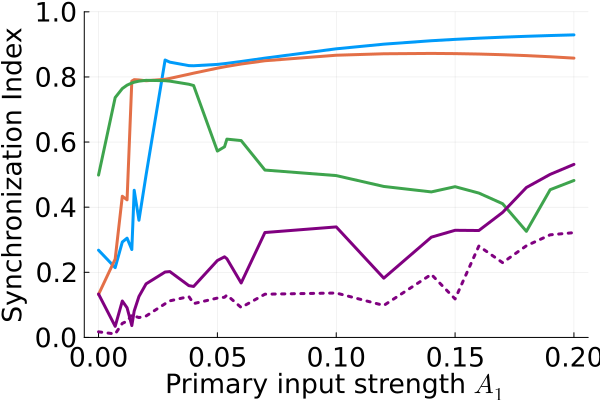} 
\end{tabular}
\caption{Synchronization index $r$ between the E-cells and the primary input $p_1$ of period $T_1$ (solid curves) and the distractor $p_2$ of period $T_2=0.8T$ (dashed curves) for the controlled system \eqref{Epop}-\eqref{Ipop} and for different values of the distractor strength $A_2$ (see legend). The periods of the primary are (a) $T_1 = T$, (b) $T_1 = 1.1T$ and (c) $T_1=1.2 T$.}
\label{fig:syncindexes}
\end{figure}

In addition to compute the SI, we also compute the amplification factor $\Delta \alpha$ defined in \eqref{eq:Dalpha} for cases in Figure \ref{fig:syncindexes} for which the SI is large enough (see Figure \ref{fig:deltaalphas}). We can observe that the network is sensitive to the input as long as we have a high enough SI with the primary. Thus, our control strategy is able to establish selective communication in the mean-field models for populations of neurons when the amplitude of the distractor remains reasonable.

\begin{figure}
\begin{tabular}{lll}
(a) & (b) & (c) \\
\includegraphics[width=0.33\textwidth]{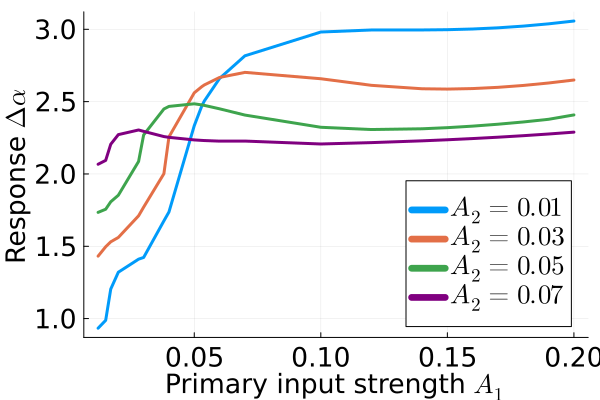} &
\includegraphics[width=0.33\textwidth]{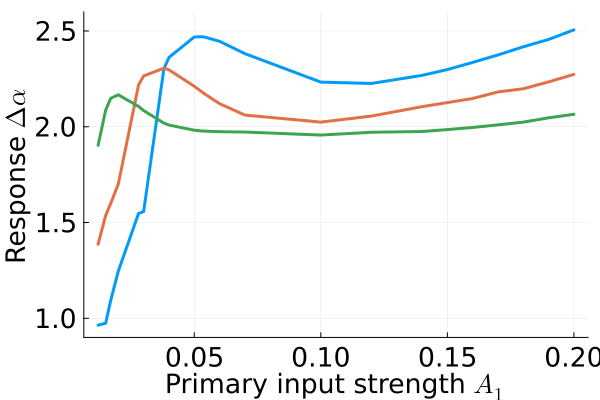} &
 \includegraphics[width=0.33\textwidth]{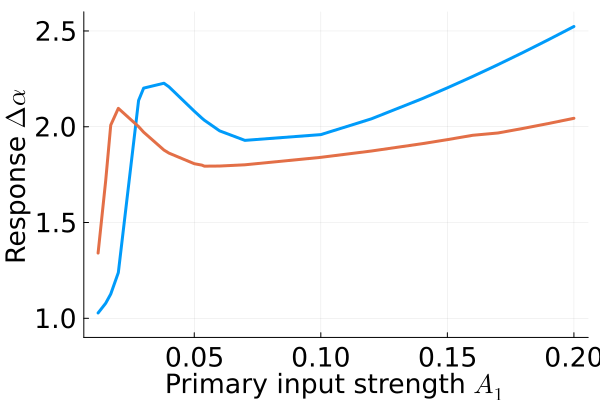} \\
 \end{tabular}
\caption{Amplification of the response of the target system due to changes in the strength $A_1$ of the primary input. Factor $\Delta\alpha$ for the controlled system as a function of the primary input strength $A_1$ in the presence of a distractor $p_2$ with $T_2=0.8T$ for different values of the distractor strength $A_2$ (see legend) and different periods of the primary (a) $T_1 = T$ (b) $T_1 = 1.1T$ (c) $T_1=1.2T$. The factor $\Delta \alpha$ has been computed only in those cases in which the SI index is above 0.6, see Figure \ref{fig:syncindexes}.}
\label{fig:deltaalphas}
\end{figure}

\section{Discussion}\label{sec:discussion}

In this paper we have studied how we can design an optimal-control strategy to control the phase of an oscillator in the context of communication through coherence (CTC) theory. To do so, we have designed a control strategy based on Pontryagin's Maximum Principle, involving a description of the dynamics using phase and amplitude variables, which guarantee a better control of the phase dynamics when the trajectory is displaced away from the limit cycle. 

Specifically, we have designed a mathematical setting to explain how a top-down input, represented by the control, can modify the dynamics of an oscillating postsynaptic group so that its oscillations synchronize with a given presynaptic input. As a result, we provide an explanation of how communication can be established, even in those cases in which the frequency of the presynaptic oscillatory input is not suitable for communication \cite{ReynerHuguet22}. We stress here that we interpret communication as the increase in the firing rate of the E-cells of the target network due to changes in the input strength.

We emphasize that we have presented novel theoretical results regarding the controllability of a system close to a limit cycle (see Proposition \ref{LocalCtrl}, which provides sufficient conditions to guarantee this controllability). Before designing the control for our system of interest (E-I network system \eqref{Epop}-\eqref{Ipop}) we have checked that it satisfies the hypothesis of Proposition \ref{LocalCtrl}. Moreover, to illustrate the relevance of our result we have applied it to other models in neuroscience having oscillatory dynamics, for which it is usual to apply phase-control techniques \cite{Moehlis2013}.

Optimal-control strategies for single neuron models (or a small number of neurons) have been previously investigated in \cite{Brown04,MoehlisT,Moehlis}, where the authors proposed a procedure for determining optimal control based on the phase or the first-order phase-amplitude reduction. In this work, we build on these methods by considering higher-order terms in the amplitude variable, allowing us to treat the problem beyond the weak coupling approximation. Additionally, we adopt a Hamiltonian formulation for the optimal-control problem based on Pontryagin's Maximum Principle, in contrast to the Lagrangian formulation \cite{MoehlisT} or the Hamilton-Jacobi-Bellman approach used in \cite{Moehlis2013}.

In the practical application of the control problem \eqref{OCP2} to the E-I network system, we emphasize the utilization of the linear approximation in $\sigma$ for the terms $Z^v(\theta,\sigma)$ and $I^v(\theta,\sigma)$. This choice ensures more precise results for the phase dynamics.  In fact, one could obtain even higher accuracy, especially when the trajectory deviates significantly from the original oscillator, by using a higher order approximation for the aforementioned terms, following the methodology presented in \cite{Perez_Cervera20}. 

We point out that the control has been optimized for a single cycle and then applied periodically throughout the $N$ cycles of the full simulation. Note that the initial conditions of each cycle are different from the first one and so our approach hinders the control from remaining optimal for the entire duration of the simulation. 
Instead, we could have optimized its action on the complete simulation, that is, solving the \eqref{OCP2} problem for a final time $N\,t_f$, which would provide an optimal result; however, it would depend on the number of cycles thus compromising the clarity of the exposition. Since our purpose was providing a proof of concept of the application of control theory for this problem, here we have preferred to use the suboptimal approach. The development of an optimal control that ensures optimality for the full simulation is left for future work.

We also stress that applying the same control over multiple periods displaces the trajectory away from the original oscillator. For this reason, we have not penalized the distance to the limit cycle in the cost function. If an optimal control were to be designed for the full simulation, as suggested in the previous paragraph, displacement from the original oscillator could be avoided by imposing additional constraints on the control function, such as including the amplitude penalization in the cost function (controlled by the parameter $\alpha$ in equation~\eqref{eq:cost2}). 

We emphasize that the models considered for CTC strike a balance between realism and the ability to draw insights from experimental findings. Additionally, we intentionally designed the study to be applicable to a broad range of brain regions rather than focusing on specific ones. Finally, we want to highlight that our methodology has multiple applications beyond the field of computational neuroscience, which merit exploration in the future.

\section*{Acknowledgements}
Work produced with the support of the grant PID-2021-122954NB-I00 (MO, AG, GH) and PID-2022-137708NB-I00 (AG) funded by MCIN/AEI/ 10.13039/501100011033 and “ERDF: A way of making Europe” and the Maria de Maeztu Award for Centers and Units of Excellence in R\&D (CEX2020-001084-M). 
Authors want to thank Alberto P\'erez Cervera (UCM) for providing support with the numerical code for the Phase-Amplitude reduction. We also acknowledge the use of the UPC Dynamical Systems group’s cluster for research computing \texttt{https://dynamicalsystems.upc.edu/en/computing/}.

\section*{Appendix} \label{sec:Appendix}

In this section we apply the results of Proposition \ref{LocalCtrl} to some classical single cell models in neuroscience, to show local controllability around the limit cycle. We first present the analytical proof of local controllability for the 2D FitzHugh-Nagumo and Morris-Lecar models and later a numerical evidence for local controllability of the limit cycle for the classical Hodgkin-Huxley model.

Let us first consider the controlled FitzHugh–Nagumo \cite{Fitzhugh_61,Nagumo_1962}
\begin{equation}\label{eq:FNc}
    \begin{cases}
    \dot{V}=V-\frac{V^3}{3}-w+u\\
    \dot{w}=\epsilon(V-a-bw)
    \end{cases}
\end{equation} 
and the controlled Morris-Lecar model \cite{ML81},
\begin{equation}\label{eq:MLc}
\begin{cases}
\dot{V}=-g_L\,(V-V_L)-g_{Ca}\, m(V)(V-V_{Ca})-g_K\, w\,(V-V_K)+ I_{app} + u\\
\dot{w}=\frac{w_s(V)-w}{\tau_w(V)},
\end{cases}
\end{equation}
with \[m(V) = 1/2(1+\tanh(\frac{V-V_1}{V_2}),\] 
\[w_s(V)=1/2(1+\tanh(\frac{V-V_3}{V_4})),\] and 
\[\tau_w(V)=\dfrac{1}{\varphi\cosh(\frac{V-V_3}{2V_4})}.\]

\begin{corollary}
The FitzHugh–Nagumo model \eqref{eq:FNc}
and the Morris-Lecar model \eqref{eq:MLc} are controllable in a neighborhood of their periodic orbits (whenever the parameters allow such orbits).
\end{corollary}
\begin{proof}
For the FitzHugh-Nagumo model \eqref{eq:FNc} we have $F_0=\begin{pmatrix}V-\frac{V^3}{3}-w\\
\epsilon(V-a-bw)
\end{pmatrix}
$
and $F_1=\begin{pmatrix} 1 \\
0
\end{pmatrix}.$
Taking the first Lie bracket is enough: $[F_0,F_1]=\begin{pmatrix} 1-V^2 \\
\epsilon
\end{pmatrix}.$ Thus, the hypothesis of Proposition \ref{LocalCtrl} is satisfied as long as $\epsilon\neq 0$.

For the Morris-Lecar model \eqref{eq:MLc} we have \[F_0(V,w)=\begin{pmatrix}-g_{Ca} m(V)(V-V_{Ca})- g_K w(V-V_K) - g_L(V-V_L)\\
\dfrac{w_s(V)-w}{\tau_w(V)}
\end{pmatrix},\]
and
\[F_1=\begin{pmatrix}1\\
0
\end{pmatrix}.\]
We compute 
\[[F_0,F_1]=\begin{pmatrix}
g_{Ca} m'_V(V)(V-V_{Ca})+ g_{Ca} m(V)+g_Kw+ g_L\\
\\
\dfrac{w'_s(V)\tau_w(V)-(w_s(V)-w)\tau'_w(V)}{\tau^2_w(V)}
\end{pmatrix}.\] 
Notice that the vectors $F_1$ and $[F_0,F_1]$ are enough to generate linearly the tangent space except when $w'_s(V)\tau_w(V)-(w_s(V)-w)\tau'_w(V)=0$. But this implies $w_s(V)-w=c\tau_w(V)$ for some constant $c$, which cannot be true given their definitions, or on $w=w_s(V)$ and $w'_s(V)=0$, which never occurs given the definition of $w_s(V)$. Thus we have concluded the proof.
\end{proof}

\begin{figure}
\begin{tabular}{ll}
(a) & (b) \\
\includegraphics[scale=0.3]{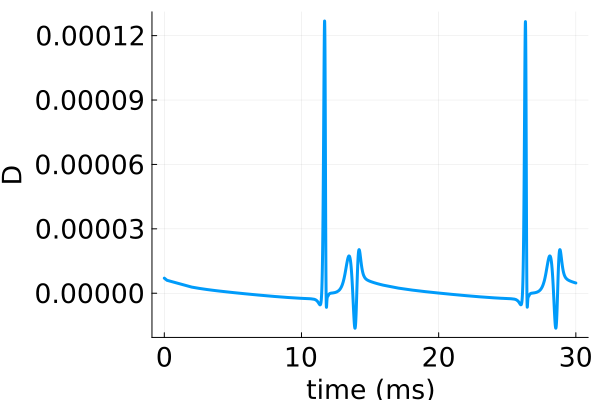} &
\includegraphics[scale=.3]{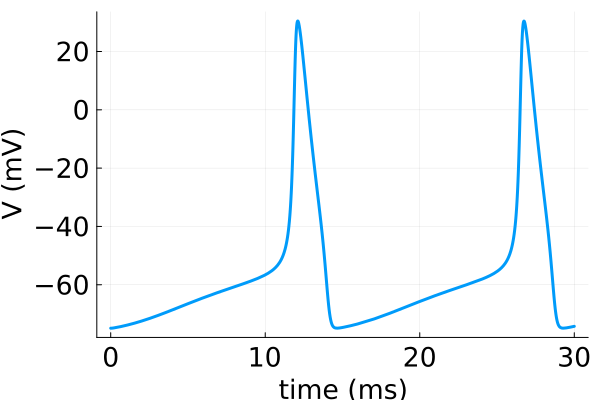} \\
\end{tabular}
\caption{(a) Determinant of the matrix defined in \eqref{matrixcont} and (b) Voltage coordinate along the periodic orbit $\Gamma$ of the Hodgkin-Huxley model \eqref{HHmodel} with the following set of parameters: 
$C = 1$ ($\mu\mbox{F}/\mbox{cm}^2$), 
$g_{Na} = 120,
g_K = 36,
g_L = 0.3$ ($\mbox{mS}/\mbox{cm}^2$),                                                
$V_{Na} = 40,
V_K = -77,
V_L = -54.4$ (mV),
$I = 10$  ($\mu\mbox{A}/\mbox{cm}^2$).
}
\label{fig:dual}
\end{figure}

For higher dimensional models like the Hodgkin-Huxley model \cite{HH}, the algebraic expressions become more complicate. In spite of this, one can numerically check that the Lie brackets with the drift generate the whole tangent space at least in one point of the periodic orbit. 

The following numerical analysis provides local controllability for the classical Hodgkin-Huxley (HH) model, for a given set of the parameters. 
The dynamical system writes as 
\begin{equation}\label{HHmodel}
\begin{cases}
    C\dot{V}= -g_L(V-V_L)-g_K n^4(V-V_K)-g_{Na}m^3 h(V-V_{Na})+I, \\
    \dot{n} = \alpha_n(V)(1-n)-\beta_n(V)n,\\
    \dot{m} = \alpha_m(V)(1-m)-\beta_m(V)m,\\
    \dot{h} = \alpha_h(V)(1-h)-\beta_h(V)h,\\
\end{cases}
\end{equation}
with
$$
\begin{array}{rlrl}
    \alpha_n (V)&= \displaystyle\frac{0.01\,(V+55)}{1-\exp(\frac{V+55}{10})},&
    \beta_n (V)&= 0.125\,\displaystyle\exp\left(-\frac{V+65}{80}\right);\\  & & & \\
    \alpha_m (V)&= \displaystyle\frac{0.1\,(V+40)}{1-\exp(-\frac{V+40}{10})}, &
    \beta_m (V)&= 4\,\displaystyle\exp\left(-\frac{V+65}{18}\right);\\ & & & \\
    \alpha_h (V)&= 0.07\,\displaystyle\exp\left(\frac{-(V+65)}{20}\right),&
    \beta_h (V)&= \displaystyle\frac{1}{1-\exp(-\frac{V+35}{10})}.
\end{array}
$$
Let us consider a choice of parameters such that the HH system has a limit cycle $\Gamma$ parameterized by the function $\gamma(t)$, where $t$ is the time. Let us set, 
\begin{equation}
A(t) = (F_1, [F_0,F_1], [F_0,[F_0,F_1]], [F_0,[F_0,[F_0,F_1]]])(\gamma(t)),
\end{equation}
and $D(t)=\det(A(t))$. We can numerically check that $D \not\equiv 0$ along the periodic orbit. In Figure \ref{fig:dual} we plot the function $D$ along $\Gamma$ for a particular choice of the parameters specified in the caption. Notice that the dependence in the voltage coordinate is strong, as it becomes non zero when the neuron spikes. Compare the positions of the peaks in Figures \ref{fig:dual} (a) and (b).

Thus, for this particular set of parameters, we can conclude that the model is controllable in a neighborhood of $\Gamma$ according to Proposition \ref{LocalCtrl}.

\bibliographystyle{plain}
\bibliography{references}
\end{document}